\documentclass[11pt,reqno]{article}
\usepackage{amssymb,amscd,amsmath,amsthm,color}

\newtheorem{thm}{Theorem}[section]
\newtheorem{prop}[thm]{Proposition}
\newtheorem{lemma}[thm]{Lemma}
\newtheorem{cor}[thm]{Corollary}
\newtheorem{definition}[thm]{Definition}
\newtheorem{remark}[thm]{Remark}
\newtheorem{example}[thm]{Example}

\numberwithin{equation}{section}

\def\bR{\mathbb{R}}
\def\bC{\mathbb{C}}
\def\bN{\mathbb{N}}
\def\bM{\mathbb{M}}
\def\cH{\mathcal{H}}
\def\cB{\mathcal{B}}
\def\BH{\cB(\cH)}
\def\cP{\mathcal{P}}
\def\cC{\mathcal{C}}
\def\cD{\mathcal{D}}

\def\Tr{\mathrm{Tr}\,}
\def\Sp{\mathrm{spec}\,}
\def\eps{\varepsilon}
\def\<{\langle}
\def\>{\rangle}
\def\ffi{\varphi}
\def\tr{\mathrm{tr}}
\def\Re{\mathrm{Re}\,}
\def\Im{\mathrm{Im}\,}
\def\sa{\mathrm{sa}}
\def\ext{\mathrm{ext}}
\def\fS{\frak{S}}
\def\cA{\mathcal{A}}

\begin{document}
\allowdisplaybreaks

\centerline{\LARGE Quantum $f$-divergences in von Neumann algebras I.}
\medskip
\centerline{\LARGE Standard $f$-divergences}
\bigskip
\bigskip
\centerline{\Large Fumio Hiai\footnote{{\it E-mail address:} hiai.fumio@gmail.com}}

\begin{center}
$^1$\,Tohoku University (Emeritus), \\
Hakusan 3-8-16-303, Abiko 270-1154, Japan
\end{center}

\begin{center}
{\Large Dedicated to the memory of D\'enes Petz}
\end{center}

\medskip
\begin{abstract}
We make a systematic study of standard $f$-divergences in general von Neumann algebras. An
important ingredient of our study is to extend Kosaki's variational expression of the relative
entropy to an arbitary standard $f$-divergence, from which most of the important properties of
standard $f$-divergences follow immediately. In a similar manner we give a comprehensive
exposition on the R\'enyi divergence in von Neumann algebra. Some results on relative
hamiltonians formerly studied by Araki and Donald are improved as a by-product.

\bigskip\noindent
{\it Keywords and phrases:}
Standard $f$-diveregence, relative entropy, R\'enyi divergence, relative modular operator,
standard form, operator convex function, Haagerup's $L^p$-space, relative hamiltonian.

\bigskip\noindent
{\it Mathematics Subject Classification 2010:} 81P45, 81P16, 46L10, 46L53, 94A17
\end{abstract}

\section{Introduction}

The notion of quantum divergences is among the most significant ones in quantum information
theory, with various applications, in particular, to defining important quantum quantities to
descriminate between states of a quantum system. A quantum system is mathematically described,
in most cases, by an operator algebra $\cA$ on a Hilbert space (either finite-dimensional or
infinite-dimensional), and a quantum divergence is generally given as a function
$S(\rho\|\sigma)$ of two states (or more generally, two positive linear functionals) $\rho$
and $\sigma$ on $\cA$. Among various quantum divergences, the most notable is the
\emph{relative entropy} having a long history as the quantum version of the Kullback-Leibler
divergence in classical theory. Indeed, the relative entropy $D(\rho\|\sigma)$ was first
introduced in 1962 by Umegaki \cite{Um} for normal states $\rho,\sigma$ on a semifinite von
Neumann algebra $M$ as follows:
\begin{align}\label{F-1.1}
D(\rho\|\sigma):=\begin{cases}
\tau(d_\rho(\log d_\rho-\log d_\sigma)) & \text{if $s(\rho)\le s(\sigma)$}, \\
+\infty & \text{otherwise},
\end{cases}
\end{align}
where $\tau$ is a semifinite trace on $M$, $d_\rho$ is the density operator of $\rho$ with
respect to $\tau$ and $s(\rho)$ is the support projection of $\rho$. Later in 1970's Araki
\cite{Ar5,Ar2} extended Umegaki's relative entropy, by introducing the \emph{relative modular
operator} $\Delta_{\rho,\sigma}$ for normal states $\rho,\sigma$, to general
von Neumann algebras as
\begin{align}\label{F-1.2}
D(\rho\|\sigma):=\begin{cases}
-\<\xi_\rho,(\log\Delta_{\sigma,\rho})\xi_\rho\>
=\<\xi_\sigma,(\Delta_{\rho,\sigma}\log\Delta_{\rho,\sigma})\xi_\sigma\>
& \text{if $s(\rho)\le s(\sigma)$}, \\
+\infty & \text{otherwise},
\end{cases}
\end{align}
where $\xi_\rho$ is the vector representative of $\rho$ in the standard representation of $M$
(see Section 2.1 below). A remarkable progress on the relative entropy was made when
Kosaki \cite{Ko} gave a variational expression of $D(\rho\|\sigma)$ and showed that all
important properties of $D(\rho\|\sigma)$ immediately follow from the expression.

A more general form of quantum divergences was considered by Kosaki \cite{Ko0} to generalize
the Wigner-Yanase-Dyson-Lieb concavity, and later discussed in more detail by Petz
\cite{Pe,Pe0} with name \emph{quasi-entropy}. The reader can refer to \cite{OP} for details
on the relative entropy and quasi-entropies. The \emph{standard $f$-divergences}
$S_f(\rho\|\sigma)$ studied in, e.g., \cite{HMPB,HiMo} in the finite-dimensional case are a
special case of quasi-entropies but a natural class of quantum divergences generalizing the
classical $f$-divergences. A significant property satisfied by quantum divergences mentioned
above is the monotonicity property, that is, the inequality
\begin{align}\label{F-1.3}
S(\rho\circ\Phi\|\sigma\circ\Phi)\le S(\rho\|\sigma)
\end{align}
under positive linear maps $\Phi:\cB\to\cA$ between operator algebras assumed to be
unit-preserving and completely positive (or more weakly a Schwarz map). An important issue in
connection with this property is to prove the \emph{reversibility} of $\Phi$, i.e., the
existence of a recovery map $\Psi:\cA\to\cB$ satisfying $\rho\circ\Phi\circ\Psi=\rho$ and
$\sigma\circ\Phi\circ\Psi=\sigma$ under the equality condition in the monotonicity inequality
in \eqref{F-1.3}. In the special case where $\cB$ is a subalgebra of $\cA$ and $\Phi$ is the
injection, the reversibility of $\Phi$ on $\{\rho,\sigma\}$ is called the \emph{sufficiency}
of $\cB$ for $\{\rho,\sigma\}$. This line of research in von Neumann algebras was initiated
by Petz, e.g., \cite{Pe1,Pe3,JP} for the relative entropy and the transition probability
(i.e., the R\'enyi divergence with parameter $1/2$). The extension of the reversibility to
more general standard $f$-divergence, though in the finite-dimensional case, has been done
in \cite{HMPB,Je0,HiMo}.

Our aim in this paper is to propose a new approach to the theory of standard $f$-divergences
in general von Neumann algebras. For this, in Section 2 we first give the definition and some
basic properties of the standard $f$-divergence $S_f(\rho\|\sigma)$ of normal positive linear
functionals $\rho,\sigma$ on a von Neumann algebra, when $f$ is a general convex function on
$(0,\infty)$. The idea of the definition is essentially the same as \eqref{F-1.2}, based on
the relative modular operator $\Delta_{\rho,\sigma}$, but without  any assumption on the
boundary values of $f(t)$ at $t$ zero and infinity. For further discussions we assume that
$f$ is an operator convex function on $(0,+\infty)$. In Section 3 we give a variational
expression of $S_f(\rho\|\sigma)$ by utilizing the integral expression of $f$ and modifying
Kosaki's expression \cite{Ko} of the relative entropy. Our variational expression is a bit
different from Kosaki's one even in the case of the relative entropy. Next in Section 4, we
present various properties of $S_f(\rho\|\sigma)$ such as monotonicity property in
\eqref{F-1.3}, joint lower semicontinuity, joint convexity, etc.\ as straightforward
consequences from the variational expression. In this way, we can study the standard
$f$-divergence in von Neumann algebras along a very streamlined track, which is a special
feature of our presentation.

The quantum \emph{R\'enyi divergence} $D_\alpha(\rho\|\sigma)$ with parameter
$\alpha\in[0,+\infty)\setminus\{1\}$ is of quite use in quantum information as a quantum
version of the classical R\'enyi divergence. In the finite-dimensional (or the matrix) case,
the R\'enyi divergence is defined by
$$
D_\alpha(\rho\|\sigma):={1\over\alpha-1}\log{Q_\alpha(\rho\|\sigma)\over\Tr\rho},
$$
where $Q_\alpha(\rho\|\sigma):=\Tr(\rho^\alpha\sigma^{1-\alpha})$ that is essentially the
standard $f$-divergence with $f(t)=t^\alpha$. So $D_\alpha$ is indeed a variant of standard
$f$-divergences. In these years it has also been widely known that another type of quantum
R\'enyi divervgence, called the \emph{sandwiched R\'enyi divergence} and denoted by
$\widetilde D_\alpha(\rho\|\sigma)$, is equally useful in quantum information, in particular,
in quantum state discrimination, see \cite{MO} for example. The definition of
$\widetilde D_\alpha$ is similar to $D_\alpha$ by replacing $Q_\alpha(\rho\|\sigma)$ with
$\widetilde Q_\alpha(\rho\|\sigma):=\Tr\bigl((\sigma^{(1-\alpha)/2\alpha}\rho
\sigma^{(1-\alpha)/2\alpha})^\alpha\bigr)$, though $\widetilde Q_\alpha$ is no longer in the
class of standard $f$-divergences. Moreover, the so-called \emph{$\alpha$-$z$-R\'enyi
divergence} introduced in \cite{AD} is a two-parameter common generalization of $D_\alpha$
and $\widetilde D_\alpha$. Motivated by the current situation of quantum R\'enyi divergences,
the authors in \cite{BST,Je1,Je2} have recently extended the sandwiched version
$\widetilde D_\alpha$ to the von Neumann algebra setting. Indeed, in these papers, the
quantity $\widetilde Q_\alpha$ is defined in von Neumann algebras by using Araki and Masuda's
$L^p$-spaces \cite{AM} or Haagerup's $L^p$-spaces \cite{Ha2,Te}. Although those papers contain
some discussions on $D_\alpha$ as well, it seems that the expositions on $D_\alpha$ there are
not comprehensive. Thus, in Section 5 we present a thorough exposition on the R\'enyi
divergence in von Neumann algebras, while it is more or less specialization of the results
of Section 4.

This paper has two appendices. In Appendix A we give a brief survey on Haagerup's $L^p$-spaces
and the description of the R\'enyi divergence in terms of them. In Appendix B we revisit the
former results in \cite{Ar3,Do} on relative hamiltonians and their relation to the relative
entropy, and improve them based on Haagerup's $L^p$-spaces and the fact that
$D=\lim_{\alpha\nearrow1}D_\alpha$.

\section{Definition of standard $f$-divergences}

\subsection{Relative modular operators}

Let $M$ be a general von Neumann algebra, and $M_*^+$ be the positive cone of the predual
$M_*$ consisting of normal positive linear functionals on $M$. Throughout the paper, we
consider $M$ in its \emph{standard form} $(M,\cH,J,\cP)$, that is, $M$ is represented on
a Hilbert space $\cH$ with a conjugate-linear involution $J$ and a self-dual cone $\cP$
called the \emph{natural cone}, for which the following hold:
\begin{itemize}
\item[(1)] $JMJ=M'$,
\item[(2)] $JxJ=x^*$,\quad$x\in M\cap M'$ (the center of $M$),
\item[(3)] $J\xi=\xi$,\quad$\xi\in\cP$,
\item[(4)] $xJxJ\cP\subset\cP$,\quad$x\in M$.
\end{itemize}
Recall \cite{Ha} that any von Neumann algebra has a standard form, which is unique in the
sense that if $(M,\cH,J,\cP)$ and $(\tilde M,\tilde\cH,\tilde J,\tilde\cP)$ are two standard
forms and $\Phi:M\to\tilde M$ is a $*$-isomorphism, then there is a unique unitary
$u:\cH\to\tilde\cH$ such that $\Phi(x)=uxu^*$ for $x\in M$, $\tilde J=uJu^*$ and
$\tilde\cP=u\cP$. See \cite{Ha} (also \cite{Ar1,Co}) for more details on the standard form.

Every $\sigma\in M_*^+$ has a unique \emph{vector representative} $\xi_\sigma\in\cP$ so that
$\sigma(x)=\<\xi_\sigma,x\xi_\sigma\>$, $x\in M$. We have the support projections
$s_M(\sigma)\in M$ and $s_{M'}(\sigma)\in M'$ of $\sigma$, that is, $s_M(\sigma)$ is the
orthogonal projection onto $\overline{M'\xi_\sigma}$ and $s_{M'}(\sigma)$ is that onto
$\overline{M\xi_\sigma}$. Note that $s_{M'}(\sigma)=Js_M(\sigma)J$.

For each $\rho,\sigma\in M_*^+$ the operators $S_{\rho,\sigma}$ and $F_{\rho,\sigma}$ are
defined by
$$
S_{\rho,\sigma}(x\xi_\sigma+\eta):=s_M(\sigma)x^*\xi_\rho,\qquad
x\in M,\ \eta\in(1-s_{M'}(\sigma))\cH,
$$
$$
F_{\rho,\sigma}(x'\xi_\sigma+\zeta):=s_{M'}(\sigma)x^{\prime*}\xi_\rho,\qquad
x'\in M',\ \zeta\in(1-s_M(\sigma))\cH.
$$
Then $S_{\rho,\sigma}$ and $F_{\rho,\sigma}$ are closable conjugate-linear operators such
that $S_{\rho,\sigma}^*=\overline F_{\rho,\sigma}$, see \cite[Lemma 2.2]{Ar2}. The
\emph{relative modular operator} $\Delta_{\rho,\sigma}$ introduced in \cite{Ar2} is
$$
\Delta_{\rho,\sigma}:=S_{\rho,\sigma}^*\overline S_{\rho,\sigma},
$$
and the polar decomposition of $\overline S_{\rho,\sigma}$ is given as
\begin{align}\label{F-2.1}
\overline S_{\rho,\sigma}=J\Delta_{\rho,\sigma}^{1/2}.
\end{align}
Recall that the support projection of $\Delta_{\rho,\sigma}$ is $s_M(\rho)s_{M'}(\sigma)$.
We write the spectral decomposition of $\Delta_{\rho,\sigma}$ as
\begin{align}\label{F-2.2}
\Delta_{\rho,\sigma}=\int_{[0,+\infty)}t\,dE_{\rho,\sigma}(t).
\end{align}
When $\rho=\sigma$, $\Delta_{\sigma,\sigma}$ is the \emph{modular operator} $\Delta_\sigma$
for $\sigma$. Note that $\xi_\sigma$ is an eigenvector of $\Delta_\sigma$ with eigenvalue $1$,
so that $\Delta_\sigma\xi_\sigma=\xi_\sigma$.

\subsection{Standard $f$-divergence $S_f(\rho\|\sigma)$}

Let $f:(0,+\infty)\to\bR$ be a convex function. Then the limits
$$
f(0^+):=\lim_{t\searrow0}f(t),\qquad f'(+\infty):=\lim_{t\to+\infty}{f(t)\over t}
$$
exist in $(-\infty,+\infty]$. Below we will understand the expression $bf(a/b)$ for $a=0$ or
$b=0$ in the following way:
\begin{align}\label{F-2.3}
bf(0/b):=f(0^+)b\quad\mbox{for $b\ge0$},\quad
0f(a/0):=\lim_{t\searrow0}tf(a/t)=f'(+\infty)a\quad\mbox{for $a>0$},
\end{align}
where we use the convention that $(+\infty)0:=0$ and $(+\infty)c:=+\infty$ for $c>0$. In
particular, we fix $0f(0/0)=0$.

The next definition is a specialization of the quasi-entropy \cite{Ko0,Pe} with
modifications.

\begin{definition}\label{D-2.1}\rm
For each $\rho,\sigma\in M_*^+$, with the spectral decomposition in \eqref{F-2.2}, we define
the self-adjoint operator $f(\Delta_{\rho,\sigma})$ on $s_M(\rho)s_{M'}(\sigma)\cH$ by
\begin{align}\label{F-2.4}
f(\Delta_{\rho,\sigma}):=\int_{(0,+\infty)}f(t)\,dE_{\rho,\sigma}(t),
\end{align}
and define
\begin{align}\label{F-2.5}
\<\xi_\sigma,f(\Delta_{\rho,\sigma})\xi_\sigma\>
:=\int_{(0,+\infty)}f(t)\,d\|E_{\rho,\sigma}(t)\xi_\sigma\|^2.
\end{align}
We then introduce the {\it standard $f$-divergence} $S_f(\rho\|\sigma)$ of $\rho,\sigma$ by
\begin{align}\label{F-2.6}
S_f(\rho\|\sigma):=\<\xi_\sigma,f(\Delta_{\rho,\sigma})\xi_\sigma\>
+f(0^+)\sigma(1-s_M(\rho))+f'(+\infty)\rho(1-s_M(\sigma)).
\end{align}
\end{definition}

Note that the integral in \eqref{F-2.4} is on $(0,+\infty)$ instead of $[0,+\infty)$. The
left-hand expression of \eqref{F-2.5} should be understood, to be precise, in the sense of a
lower-bounded form (see \cite{RS}), which equals the integral in the right-hand side. We first
give a lemma to justify the above definition.

\begin{lemma}\label{L-2.2}
For every $\rho,\sigma\in M_*^+$, $S_f(\rho\|\sigma)$ is well defined with values in
$(-\infty,+\infty]$.
\end{lemma}

\begin{proof}
From the convexity of $f$, there are $a,b\in\bR$ such that $f(t)\ge a+bt$ for all
$t\in(0,+\infty)$. We have
\begin{align*}
\int_{(0,+\infty)}f(t)\,d\|E_{\rho,\sigma}(t)\xi_\sigma\|^2
&\ge\int_{(0,+\infty)}(a+bt)\,d\|E_{\rho,\sigma}(t)\xi_\sigma\|^2 \\
&=a\|s_M(\rho)s_{M'}(\sigma)\xi_\sigma\|^2+b\|\Delta_{\rho,\sigma}^{1/2}\xi_\sigma\|^2 \\
&=a\|s_M(\rho)\xi_\sigma\|^2+b\|s_M(\sigma)\xi_\rho\|^2 \\
&=a\sigma(s_M(\rho))+b\rho(s_M(\sigma))>-\infty,
\end{align*}
since $J\Delta_{\rho,\sigma}^{1/2}\xi_\sigma=s_M(\sigma)\xi_\rho$ by \eqref{F-2.1}. Hence
$S_f(\rho\|\sigma)\in(-\infty,+\infty]$.
\end{proof}

By the above proof and \eqref{F-2.6} we also see that
\begin{align}\label{F-2.7}
S_{a+bt}(\rho\|\sigma)=a\sigma(1)+b\rho(1).
\end{align}

The following are some basic properties of $S_f(\rho\|\sigma)$:

\begin{prop}\label{P-2.3}
Let $\rho,\sigma\in M_*^+$.
\begin{itemize}
\item[(1)] If $\Phi:\tilde M\to M$ is a $*$-isomorphism between von Neumann algebras, then
$$
S_f(\rho\circ\Phi\|\sigma\circ\Phi)=S_f(\rho\|\sigma).
$$
\item[(2)] In the case $\rho=0$ or $\sigma=0$ or $\rho=\sigma$,
$$
S_f(0\|\sigma)=f(0^+)\sigma(1),\quad S_f(\rho\|0)=f'(+\infty)\rho(1),\quad
S_f(\sigma\|\sigma)=f(1)\sigma(1).
$$
\item[(3)] \emph{Homogeneity:} For every $\lambda\in[0,+\infty)$,
\begin{align}\label{F-2.8}
S_f(\lambda\rho\|\lambda\sigma)=\lambda S_f(\rho\|\sigma).
\end{align}
\item[(4)] \emph{Additivity:} Let $M=M_1\oplus M_2$ be the direct sum of von Neumann algebras
$M_1$ and $M_2$. If $\rho_i,\sigma_i\in(M_i)_*^+$ for $i=1,2$, then
$$
S_f(\rho_1\oplus\rho_2\|\sigma_1\oplus\sigma_2)
=S_f(\rho_1\|\sigma_1)+S_f(\rho_2\|\sigma_2).
$$
\end{itemize}
\end{prop}

\begin{proof}
(1) is clear from the uniqueness of the standard form stated in Section 2.1.

(2) is seen directly from definition \eqref{F-2.6}.

(3)\enspace
For $\lambda>0$, since $\xi_{\lambda\sigma}=\sqrt\lambda\,\xi_\sigma$ and
$\Delta_{\lambda\rho,\lambda\sigma}=\Delta_{\rho,\sigma}$, equality \eqref{F-2.8} follows.
For $\lambda=0$ both sides are zero from (2).

(4)\enspace
Note that the standard form of $M$ is the direct sum of the standard forms of $M_i$, $i=1,2$.
For $\rho:=\rho_1\oplus\rho_2$ and $\sigma:=\sigma_1\oplus\sigma_2$ we have
$\xi_\rho=\xi_{\rho_1}\oplus\xi_{\rho_2}$, $\xi_\sigma=\xi_{\sigma_1}\oplus\xi_{\sigma_2}$
and $\Delta_{\rho,\sigma}=\Delta_{\rho_1,\sigma_1}\oplus\Delta_{\rho_2,\sigma_2}$, from which
the result immediately follows.
\end{proof}

The additivity in (4) above will be improved in Section 4 (see Corollary \ref{C-4.4}\,(2)).

In our definition of $S_f(\rho\|\sigma)$ the parameter function $f$ is a convex function on
$(0,+\infty)$, not on $[0,+\infty)$. This is reasonable as the next proposition shows that
$S_f(\rho\|\sigma)$ is symmetric between $\rho$ and $\sigma$ under exchanging $f$ with its
transpose $\widetilde f$ defined by
$$
\widetilde f(t):=tf(t^{-1}),\qquad t\in(0,+\infty).
$$

\begin{prop}\label{P-2.4}
For every $\rho,\sigma\in M_*^+$,
$$
S_f(\rho\|\sigma)=S_{\widetilde f}(\sigma\|\rho).
$$
\end{prop}

\begin{proof}
Since $\widetilde f(0^+)=f'(+\infty)$ and $\widetilde f'(+\infty)=f(0^+)$, it suffices to
prove that
\begin{align}\label{F-2.9}
\<\xi_\sigma,f(\Delta_{\rho,\sigma})\xi_\sigma\>
=\<\xi_\rho,\widetilde f(\Delta_{\sigma,\rho})\xi_\rho\>.
\end{align}
Recall \cite[Theorem 2.4]{Ar2} that
$$
\Delta_{\rho,\sigma}=J\Delta_{\sigma,\rho}^{-1}J
$$
together with
$$
s_M(\rho)s_{M'}(\sigma)=Js_{M'}(\rho)JJs_M(\sigma)J=Js_M(\sigma)s_{M'}(\rho)J.
$$
Hence, since $J\Delta_{\sigma,\rho}^{1/2}\xi_\rho=s_M(\rho)\xi_\sigma$, one has
\begin{align*}
\|E_{\rho,\sigma}((0,t))\xi_\sigma\|^2
&=\|JE_{\sigma,\rho}((t^{-1},+\infty))JJ\Delta_{\sigma,\rho}^{1/2}\xi_\rho\|^2
=\|E_{\sigma,\rho}((t^{-1},+\infty))\Delta_{\sigma,\rho}^{1/2}\xi_\rho\|^2 \\
&=\int_{(t^{-1},+\infty)}s\,d\|E_{\sigma,\rho}(s)\xi_\rho\|^2
=\int_{(0,t)}s^{-1}\,d\|E_{\sigma,\rho}(s^{-1})\xi_\rho\|^2,
\end{align*}
so that $d\|E_{\rho,\sigma}(t)\xi_\sigma\|^2=t^{-1}d\|E_{\sigma,\rho}(t^{-1})\xi_\rho\|^2$ on
$(0,+\infty)$. Therefore,
\begin{align*}
\int_{(0,+\infty)}f(t)\,d\|E_{\rho,\sigma}(t)\xi_\sigma\|^2
&=\int_{(0,+\infty)}f(t)t^{-1}\,d\|E_{\sigma,\rho}(t^{-1})\xi_\rho\|^2 \\
&=\int_{(0,+\infty)}tf(t^{-1})\,d\|E_{\sigma,\rho}(t)\xi_\rho\|^2,
\end{align*}
which means \eqref{F-2.9}.
\end{proof}

\begin{example}\label{E-2.5}\rm
Let $M$ be an abelian von Neumann algebra such that $M\cong L^\infty(\Omega,\mu)$, where
$(\Omega,\mathcal{A},\mu)$ is a $\sigma$-finite measure space. Let $\rho,\sigma\in M_*^+$,
which correspond to $\phi,\psi\in L^1(\Omega,\mu)^+$ so that
$\rho(x)=\int_\Omega x\phi\,d\mu$, $\sigma(x)=\int_\Omega x\psi\,d\mu$ for
$x\in L^\infty(\Omega,\mu)$. The standard form of $L^\infty(\Omega,\mu)$ is
$(L^\infty(\Omega,\mu),L^2(\Omega,\mu),\xi\mapsto\overline\xi,L^2(\Omega,\mu)^+)$, where
$x\in L^\infty(\Omega,\mu)$ is represented on $L^2(\Omega,\mu)$ as the multiplication
operator $\xi\mapsto x\xi$ for $\xi\in L^2(\Omega,\mu)$. It is straightforward to find that
$\Delta_{\phi,\phi}$ is the multiplication of $1_{\{\psi>0\}}(\phi/\psi)$, which is the
Radon-Nikodym derivative $d\rho/d\sigma$ (restricted on the support of $\sigma$) in the
classical sense. We then have
$$
S_f(\rho\|\sigma)=\int_{\{\phi>0\}\cap\{\psi>0\}}\psi f(\phi/\psi)\,d\mu
+f(0^+)\int_{\{\phi=0\}}\psi\,d\mu+f'(+\infty)\int_{\{\psi=0\}}\phi\,d\mu,
$$
which equals the classical $f$-divergence
$S_f(\phi\|\psi)=\int_\Omega\psi f(\phi/\psi)\,d\mu$ under the convention that
$\psi(\omega)f(0/\psi(\omega))=f(0^+)\psi(\omega)$ for $\psi(\omega)\ge0$ and
$0f(\phi(\omega)/0)=\lim_{t\searrow0}tf(\phi(\omega)/t)=f'(+\infty)\phi(\omega)$ for
$\phi(\omega)>0$.
\end{example}

\begin{example}\label{E-2.6}\rm
Let $M=\BH$, i.e., a factor of type I, where $\cH$ is an arbitrary Hilbert space. The standard
form of $\BH$ is $(\BH,\cC_2(\cH),J=*,\cC_2(\cH)_+)$, where $\cC_2(\cH)$ is the
Hilbert-Schmidt class with the Hilbert-Schmidt inner product, $\cC_2(\cH)_+$ is the set of
positive operators in $\cC_2(\cH)$, and the representation of $M=\BH$ on $\cC_2(\cH)$ is
the left multiplication. Then $M'$ is the right multiplication of $M=\BH$ on $\cC_2(\cH)$. For
$\rho,\sigma\in\BH_*^+$ we have the vector representatives
$D_\rho^{1/2},D_\sigma^{1/2}\in\cC_2(\cH)_+$, where $D_\rho$ and $D_\sigma$ are the positive
trace-class operators such that $\rho(X)=\Tr D_\rho X$ and $\sigma(X)=\Tr D_\sigma X$ for
$X\in\BH$. Let
$$
D_\rho=\sum_{a\in\Sp D_\rho,\,a>0}aP_a,\qquad
D_\sigma=\sum_{b\in\Sp D_\sigma,\,b>0}bQ_b
$$
be the spectral decompositions of $D_\rho,D_\sigma$, where $\sum_{a>0}$ and $\sum_{b>0}$ are
finite or countable sums, and $P_a$ and $Q_b$ are finite-dimensional orthogonal projections.
Then the relative modular operators $\Delta_{\rho,\sigma}$ on $\cC_2(\cH)$ is given as
\begin{align}\label{F-2.10}
\Delta_{\rho,\sigma}=L_{D_\rho}R_{D_\sigma^{-1}}
=\sum_{a>0,\,b>0}ab^{-1}L_{P_a}R_{Q_b},
\end{align}
where $L_{[-]}$ and $R_{[-]}$ denote the left and the right multiplications and $D_\sigma^{-1}$
is the generalized inverse of $D_\sigma$. The last expression in the above gives the spectral
decomposition of $\Delta_{\rho,\sigma}$. The proof of this is easy as follows:

If $\Phi\in P_a\cH$ and $\Omega\in Q_b\cH$ for $a,b>0$, then
\begin{align*}
S_{\rho,\sigma}(|\Phi\>\<\Omega|)
&=S_{\rho,\sigma}(b^{-1/2}|\Phi\>\<\Omega|D_\sigma^{1/2})
=b^{-1/2}s_M(\sigma)|\Omega\>\<\Phi|D_\rho^{1/2} \\
&=a^{1/2}b^{-1/2}|\Omega\>\<\Phi|, \\
F_{\rho,\sigma}(|\Omega\>\<\Phi|)
&=F_{\rho,\sigma}(b^{-1/2}D_\sigma^{1/2}|\Omega\>\<\Phi|)
=b^{-1/2}s_{M'}(\sigma)D_\rho^{1/2}|\Phi\>\<\Omega| \\
&=a^{1/2}b^{-1/2}|\Phi\>\<\Omega|s_M(\sigma)
=a^{1/2}b^{-1/2}|\Phi\>\<\Omega|,
\end{align*}
which imply that
$$
\Delta_{\rho,\sigma}(|\Phi\>\<\Omega|)=ab^{-1}|\Phi\>\<\Omega|.
$$
Since the range of $L_{P_a}R_{Q_b}$ is the span of $|\Phi\>\<\Omega|$ for $\Phi\in P_a\cH$
and $\Omega\in Q_b\cH$,
$$
\sum_{a,b>0,\,ab^{-1}=c}L_{Pa}R_{Q_b}
$$
is the spectral projection of $\Delta_{\rho,\sigma}$ corresponding to the eigenvalue $c>0$.

Moreover, it follows from \eqref{F-2.10} that the definition of $S_f(\rho\|\sigma)$ in
\eqref{F-2.6} is rewritten as
$$
S_f(\rho\|\sigma)=\sum_{a>0}\sum_{b>0}bf(ab^{-1})\Tr P_aQ_b
+f(0^+)\Tr(I-D_\rho^0)D_\sigma+f'(+\infty)\Tr D_\rho(I-D_\sigma^0),
$$
which coincides with an expression in \cite[Proposition 3.2]{HiMo} when $\dim\cH<+\infty$.
\end{example}

\begin{remark}\label{R-2.7}\rm
Let $f:[0,+\infty)\to\bR$ be a continuous function. For $\rho,\sigma\in M_*^+$ and $k\in M$,
the {\it quasi-entropy} $S_f^k(\rho\|\sigma)$ was introduced in \cite{Pe} by
\begin{align}
S_f^k(\rho\|\sigma)&:=\<k\xi_\sigma,f(\Delta_{\rho,\sigma})k\xi_\sigma\>
=\int_{[0,+\infty)}f(t)\,d\|E_{\rho,\sigma}(t)k\xi_\sigma\|^2 \nonumber\\
&\ =\int_{(0,+\infty)}f(t)\,d\|E_{\rho,\sigma}(t)k\xi_\sigma\|^2
+f(0^+)\<k\xi_\sigma,(1-s_M(\rho)s_{M'}(\sigma))k\xi_\sigma\>. \label{F-2.11}
\end{align}
Comparing \eqref{F-2.11} with \eqref{F-2.5} and \eqref{F-2.6} we note that
$$
S_f(\rho\|\sigma)=S_f^{k=1}(\rho\|\sigma)+f'(+\infty)\rho(1-s_M(\sigma)).
$$
In particular, when $M=\BH$ with $\dim\cH<+\infty$ and $f(0^+)<+\infty$, the quasi-entropy
\eqref{F-2.11} with $k=1$ has finite values for all $\rho,\sigma$, which is improper as a
standard $f$-divergence. For example, when $f(t):=t\log t$ so that $f(0^+)=0$ and
$f'(+\infty)=+\infty$, one can easily check that for $\rho=\Tr(D_\rho\,\cdot)$ and
$\sigma=\Tr(D_\sigma\,\cdot)$ with density operators $D_\rho,D_\sigma$, expression
\eqref{F-2.11} with $k=1$ is
$$
\Tr D_\rho(\log D_\rho-{\log^+}D_\sigma),
$$
where ${\log^+}t:=\log t$ ($t>0$), ${\log^+}0:=0$. On the other hand,
$S_{t\log t}(\rho\|\sigma)$ in \eqref{F-2.6} coincides with the usual \emph{relative entropy}
$$
D(\rho\|\sigma):=\begin{cases}
\Tr D_\rho(\log D_\rho-\log D_\sigma), & s_M(\rho)\le s_M(\sigma), \\
+\infty, & s_M(\rho)\not\le s_M(\sigma).
\end{cases}
$$
\end{remark}

\section{Variational expression of standard $f$-divergences}

In this section we extend the variational expression of the relative entropy given in
\cite{Ko} to standard $f$-divergences. The extended expression will be quite useful in the
next section to verify various properties of standard $f$-divergences.

Throughout this and the next sections, we assume that a function $f:(0,+\infty)\to\bR$ is
\emph{operator convex}, i.e., the operator inequality
$$
f(\lambda A+(1-\lambda)B)\le\lambda f(A)+(1-\lambda)f(B),\qquad0\le\lambda\le1
$$
holds for every invertible $A,B\in\BH^+$ of any $\cH$. Also, a function $h:(0,+\infty)\to\bR$
is said to be \emph{operator monotone} if $A\le B$ $\implies$ $h(A)\le h(B)$ for every
invertible $A,B\in\BH^+$ of any $\cH$. It is well-known that an operator monotone function
$h$ on $(0,+\infty)$ is automatically \emph{operator concave} (i.e., $-h$ is operator convex).
For general theory on operator monotone and operator convex functions, see, e.g.,
\cite{Bh,Hi}.

Recall \cite{LR} (see also \cite[Theorem 5.1]{FHR} for a more general form) that the operator
convex function $f$ has an integral expression
\begin{align}\label{F-3.1}
f(t)=a+b(t-1)+c(t-1)^2+\int_{[0,+\infty)}{(t-1)^2\over t+s}\,d\mu(s),
\qquad t\in(0,+\infty),
\end{align}
where $a,b\in\bR$, $c\ge0$ and $\mu$ is a positive measure on $[0,+\infty)$ with
\begin{align}\label{F-3.2}
\int_{[0,+\infty)}{1\over 1+s}\,d\mu(s)<+\infty,
\end{align}
and moreover $a,b,c$ and $\mu$ are uniquely determined. Letting $d:=\mu(\{0\})\ge0$ we also write
\begin{align}\label{F-3.3}
f(t)=a+b(t-1)+c(t-1)^2+d\,{(t-1)^2\over t}+\int_{(0,+\infty)}{(t-1)^2\over t+s}\,d\mu(s),
\quad t\in(0,+\infty).
\end{align}
One can easily verify that
\begin{align}
f(0^+)&=a-b+c+(+\infty)d+\int_{(0,+\infty)}s^{-1}\,d\mu(s), \label{F-3.4}\\
f'(+\infty)&=b+(+\infty)c+d+\int_{(0,+\infty)}d\mu(s). \label{F-3.5}
\end{align}

For each $n\in\bN$ we define
\begin{align}
f_n(t)&:=a+b(t-1)+c\,{n(t-1)^2\over t+n}+d\,{(t-1)^2\over t+(1/n)} \nonumber\\
&\qquad+\int_{[1/n,n]}{(t-1)^2\over t+s}\,d\mu(s),
\qquad t\in(0,+\infty). \label{F-3.6}
\end{align}
We then have

\begin{lemma}\label{L-3.1}
For each $n\in\bN$, $f_n$ is operator convex on $(0,+\infty)$, $f_n(0^+)<+\infty$,
$f_n'(+\infty)<+\infty$ and
$$
f_n(0^+)\,\nearrow\,f(0^+),\quad f_n'(+\infty)\,\nearrow\,f'(+\infty),
\quad f_n(t)\,\nearrow\,f(t)
$$
as $n\to\infty$ for all $t\in(0,+\infty)$.
\end{lemma}

\begin{proof}
By definition \eqref{F-3.6} it is immediate to see that $f$ is an operator convex function
on $(0,+\infty)$ and
\begin{align}
f_n(0^+)&=a-b+c+nd+\int_{[1/n,n]}s^{-1}\,d\mu(s), \label{F-3.7}\\
f_n'(+\infty)&=b+nc+d+\int_{[1/n,n]}d\mu(s). \label{F-3.8}
\end{align}
It follows from \eqref{F-3.2} that \eqref{F-3.7} and \eqref{F-3.8} are finite, which
increase, by the monotone convergence theorem, to \eqref{F-3.4} and \eqref{F-3.5},
respectively. Moreover, for any $t\in(0,+\infty)$, since
$$
{n(t-1)^2\over t+n}\ \nearrow\ (t-1)^2,\qquad
{(t-1)^2\over t+(1/n)}\ \nearrow\ {(t-1)^2\over t}\qquad\mbox{as $n\nearrow\infty$},
$$
we have $f_n(t)\nearrow f(f)$ from the monotone convergence theorem again.
\end{proof}

\begin{lemma}\label{L-3.2}
For every $\rho,\sigma\in M_*^+$,
$$
S_{f_n}(\rho\|\sigma)\ \nearrow\ S_f(\rho\|\sigma)\qquad\mbox{as $n\nearrow\infty$}.
$$
\end{lemma}

\begin{proof}
By Definition \ref{D-2.1},
$$
S_{f_n}(\rho\|\sigma)=f_n(0^+)\sigma(1-s(\rho))+f_n'(+\infty)\rho(1-s(\sigma))
+\int_{(0,+\infty)}f_n(t)\,d\|E_{\rho,\sigma}(t)\xi_\sigma\|^2.
$$
By Lemma \ref{L-3.1} and the monotone convergence theorem, $S_{f_n}(\rho\|\sigma)$ increases
to $S_f(\rho\|\sigma)$ as $n\nearrow\infty$.
\end{proof}

\begin{lemma}\label{L-3.3}
For each $n\in\bN$ define an operator monotone function $h_n$ on $[0,+\infty)$ by
\begin{align}\label{F-3.9}
h_n(t)=\int_{(0,+\infty)}{t(1+s)\over t+s}\,d\nu_n(s),\qquad t\in[0,+\infty),
\end{align}
where $\nu_n$ is a finite positive measure supported on $[1/n,n]$ given by
\begin{align}\label{F-3.10}
d\nu_n(s):=c(1+n)\delta_n+d(1+n)\delta_{1/n}+1_{[1/n,n]}(s){1+s\over s}\,d\mu(s)
\end{align}
with the point masses $\delta_n$ at $n$ and $\delta_{1/n}$ at $1/n$. Then $f_n$ defined in
\eqref{F-3.6} is written as
\begin{align}\label{F-3.11}
f_n(t)=f_n(0^+)+f_n'(+\infty)t-h_n(t),\qquad t\in(0,+\infty).
\end{align}
\end{lemma}

\begin{proof}
Compute
\begin{align*}
{n(t-1)^2\over t+n}&=1+nt-(1+n){t(1+n)\over t+n}, \\
{(t-1)^2\over t+(1/n)}&=n+t-(1+n){t\bigl(1+{1\over n}\bigr)\over t+{1\over n}}, \\
{(t-1)^2\over t+s}&={1\over s}+t-{1+s\over s}\cdot{t(1+s)\over t+s}.
\end{align*}
Inserting these into definition \eqref{F-3.6} one can write
\begin{align*}
f_n(t)&=\biggl(a-b+c+nd+\int_{[1/n,n]}s^{-1}\,d\mu(s)\biggr) \\
&\qquad+\biggl(b+nc+d+\int_{[1/n,n]}d\mu(s)\biggr)t \\
&\qquad-c(1+n){t(1+n)\over t+n}-d(1+n){t\bigl(1+{1\over n}\bigr)\over t+{1\over n}} \\
&\qquad-\int_{[1/n,n]}{1+s\over s}\cdot{t(1+s)\over t+s}\,d\mu(s) \\
&=f_n(0^+)+f_n'(+\infty)t-h_n(t)
\end{align*}
thanks to \eqref{F-3.7} and \eqref{F-3.8}.
\end{proof}

Now, let $L$ be a subspace of $M$ containing $1$, and assume that $L$ is dense in $M$ with
respect to the strong* operator topology. Since $h_n(0)=h_n'(+\infty)=0$, the next lemma
follows from \cite[Theorem 2.2]{Ko}.

\begin{lemma}\label{L-3.4}
Let $h_n$ be given in \eqref{F-3.9}. Then for any $\rho,\sigma\in M_*^+$,
\begin{align*}
&\int_{(0,+\infty)}h_n(t)\,d\|E_{\rho,\sigma}(t)\xi_\sigma\|^2 \\
&\quad=\inf_{x(\cdot)}\int_{[1/n,n]}\bigl\{\sigma((1-x(s))^*(1-x(s)))
+s^{-1}\rho(x(s)x(s)^*)\bigr\}(1+s)\,d\nu_n(s),
\end{align*}
where the infimum is taken over all $L$-valued (finitely many values) step functions
$x(\cdot)$ on $(0,+\infty)$.
\end{lemma}

\begin{thm}\label{T-3.5}
Let $f$ be an operator convex function on $(0,+\infty)$. For each $n\in\bN$ let $f_n(0^+)$,
$f_n'(+\infty)$ and $\nu_n$ be given in \eqref{F-3.7}, \eqref{F-3.8} and \eqref{F-3.10},
respectively. Then for every $\rho,\sigma\in M_*^+$,
\begin{align}
S_f(\rho\|\sigma)&=\sup_{n\in\bN}\sup_{x(\cdot)}
\biggl[f_n(0^+)\sigma(1)+f_n'(+\infty)\rho(1) \nonumber\\
&\qquad-\int_{[1/n,n]}\bigl\{\sigma((1-x(s))^*(1-x(s)))
+s^{-1}\rho(x(s)x(s)^*)\bigr\}(1+s)\,d\nu_n(s)\biggr], \label{F-3.12}
\end{align}
where the supremum over $x(\cdot)$ is taken over all $L$-valued step functions as the
infimum in Lemma \ref{L-3.4}.
\end{thm}

\begin{proof}
By \eqref{F-3.11} and \eqref{F-2.7} we have
\begin{align*}
S_{f_n}(\rho\|\sigma)
&=f_n(0^+)\sigma(1)+f_n'(+\infty)\rho(1)+S_{-h_n}(\rho\|\sigma) \\
&=f_n(0^+)\sigma(1)+f_n'(+\infty)\rho(1)
-\int_{(0,+\infty)}h_n(t)\,d\|E_{\rho,\sigma}(t)\xi_\sigma\|^2.
\end{align*}
By Lemma \ref{L-3.4} we hence have
\begin{align*}
S_{f_n}(\rho\|\sigma)
&=\sup_{x(\cdot)}\biggl[f_n(0^+)\sigma(1)+f_n'(+\infty)\rho(1) \\
&\qquad-\int_{[1/n,n]}\bigl\{\sigma((1-x(s))^*(1-x(s)))
+s^{-1}\rho(x(s)x(s)^*)\bigr\}(1+s)\,d\nu_n(s)\biggr].
\end{align*}
The result follows by taking $\sup_n$ of both sides of the above and using Lemma \ref{L-3.2}.
\end{proof}

\begin{example}\label{E-12}\rm
Consider $f(t)=-\log t$, whose integral expression in \eqref{F-3.1} is
$$
-\log t=-(t-1)+\int_{(0,+\infty)}{(t-1)^2\over(t+s)(1+s)^2}\,ds.
$$
Hence, in this case,
$$
a=c=d=0,\quad b=-1,\quad d\mu(s)={1\over(1+s)^2}\,ds,
$$
$$
f(0^+)=+\infty,\qquad f'(+\infty)=0.
$$
Moreover, compute
$$
f_n(0^+)=1+\int_{1/n}^n{1\over s(1+s)^2}\,ds={2\over n+1}+\log n,
$$
$$
f_n'(+\infty)=-1+\int_{1/n}^n{1\over(1+s)^2}\,ds=-{2\over n+1},
$$
$$
d\nu_n(s)=1_{[1/n,n]}(s){1\over s(1+s)}\,ds.
$$
For every $\rho,\sigma\in M_*^+$ the \emph{relative entropy} is 
$$
D(\sigma\|\rho)=S_{t\log t}(\sigma\|\rho)=S_{-\log t}(\rho\|\sigma),
$$
for which one can write expression \eqref{F-3.12} as
\begin{align*}
D(\sigma\|\rho)&=\sup_{n\in\bN}\sup_{x(\cdot)}
\biggl[\sigma(1)\log n+(\sigma(1)-\rho(1)){2\over n+1} \\
&\qquad\quad-\int_{[1/n,n]}\bigl\{\sigma((1-x(s))^*(1-x(s)))
+s^{-1}\rho(x(s)x(s)^*)\bigr\}s^{-1}\,ds\biggr].
\end{align*}
This expression is similar to but a bit different from the variational expression
\begin{align*}
D(\sigma\|\rho)&=\sup_{n\in\bN}\sup_{x(\cdot)}\biggl[\sigma(1)\log n \\
&\qquad\quad-\int_{[1/n,+\infty)}\bigl\{\sigma((1-x(s))^*(1-x(s)))
+s^{-1}\rho(x(s)x(s)^*)\bigr\}s^{-1}\,ds\biggr]
\end{align*}
in \cite[Theorem 3.2]{Ko}.
\end{example}

\begin{remark}\rm
The variational expression in \eqref{F-3.12} with the cut-off interval $[1/n,n]$ is natural
when we consider $S_f(\rho\|\sigma)$ for general operator convex functions $f$ on
$(0,+\infty)$ with no assumption on the boundary values $f(0^+)$ and $f'(+\infty)$. This is
more explicitly justified by the fact that our expression is well behaved under taking the
transpose $\widetilde f(t)=tf(t^{-1})$. Indeed, for $f$ given in \eqref{F-3.3}, the integral
expression of $\widetilde f$ is
$$
\widetilde f(t)=\tilde a+\tilde b(t-1)+\tilde c(t-1)^2+\tilde d\,{(t-1)^2\over t}
+\int_{(0,+\infty)}{(t-1)^2\over t+s}\,d\widetilde\mu(s)
$$ 
where
$$
\tilde a=a,\qquad\tilde b=a-b,\qquad\tilde c=d,\qquad\tilde d=c,\qquad
d\widetilde\mu(s)=s\,d\mu(s^{-1}).
$$
Hence one can easily find that $\widetilde f_n(0^+)=f_n'(+\infty)$,
$\widetilde f_n'(+\infty)=f_n(0^+)$, $d\widetilde\nu_n(s)=d\nu_n(s^{-1})$, and the
expression inside the bracket $[\cdots]$ of \eqref{F-3.12} for
$S_{\widetilde f}(\rho\|\sigma)$ is
\begin{align*}
&\widetilde f_n(0^+)\sigma(1)+\widetilde f_n'(+\infty)\rho(1) \\
&\quad-\int_{[1/n,n]}\bigl\{\sigma((1-x(s))^*(1-x(s)))
+s^{-1}\rho(x(s)x(s)^*)\bigr\}(1+s)\,d\widetilde\nu_n(s) \\
&=f_n(0^+)\rho(1)+f_n'(+\infty)\sigma(1) \\
&\quad-\int_{[1/n,n]}\bigl\{\rho((1-y(s))^*(1-y(s)))
+s^{-1}\sigma(y(s)y(s)^*)\bigr\}(1+s)\,d\nu_n(s),
\end{align*}
where $y(s):=1-x(s^{-1})^*$. In this way, the variational expression in \eqref{F-3.12} enjoys
complete invariance under exchanging $(f,\rho,\sigma)$ with $(\widetilde f,\sigma,\rho)$.
\end{remark}

\section{Properties of standard $f$-divergences}

As in \cite{Ko} where the relative entropy was treated, most of the important properties of
standard $f$-divergences can immediately be verified from the variational expression in
Theorem \ref{T-3.5}.

\begin{thm}\label{T-4.1}
Let $f$ be an operator convex function on $(0,+\infty)$. Let
$\rho,\sigma,\rho_i,\sigma_i\in M_*^+$ for $i=1,2$.
\begin{itemize}
\item[(i)] \emph{Joint lower semicontinuity:} The map
$(\rho,\sigma)\in M_*^+\times M_*^+\mapsto S_f(\rho\|\sigma)\in(-\infty,+\infty]$ is jointly
lower semicontinuous in the $\sigma(M_*,M)$-topology.
\item[(ii)] \emph{Joint convexity:} The map in (i) is jointly convex and jointly subadditive,
i.e., for every $\rho_i,\sigma_i\in M_*^+$, $1\le i\le k$,
$$
S_f\Biggl(\sum_{i=1}^k\rho_i\Bigg\|\sum_{i=1}^k\sigma_i\Biggr)
\le\sum_{i=1}^kS_f(\rho_i\|\sigma_i).
$$
\item[(iii)] If $f(0^+)\le0$ and $\sigma_1\le\sigma_2$, then
$S_f(\rho\|\sigma_1)\ge S_f(\rho\|\sigma_2)$. Also, if $f'(+\infty)\le0$ and
$\rho_1\le\rho_2$, then $S_f(\rho_1\|\sigma)\ge S_f(\rho_2\|\sigma)$.
\item[(iv)] \emph{Monotonicity:} Let $N$ be another von Neumann algebra and $\Phi:N\to M$ be
a unital positive linear map that is normal (i.e., if $\{x_\alpha\}$ is an increasing net in
$M_+$ with $x_\alpha\nearrow x\in M_+$, then $\Phi(x_\alpha)\nearrow\Phi(x)$) and is a Schwarz
map (i.e., $\Phi(x^*x)\ge\Phi(x)^*\Phi(x)$ for all $x\in N$). Then
\begin{align}\label{F-4.1}
S_f(\rho\circ\Phi\|\sigma\circ\Phi)\le S_f(\rho\|\sigma).
\end{align}
In particular, if $N$ is a unital von Neumann subalgebra of $M$, then
\begin{align}\label{F-4.2}
S_f(\rho|_N\|\sigma|_N)\le S_f(\rho\|\sigma).
\end{align}
\item[(v)] \emph{Martingale convergence:} If $\{M_\alpha\}$ is an increasing net of unital
von Neumann subalgebras of $M$ such that $\bigl(\bigcup_\alpha M_\alpha\bigr)''=M$, then
$$
S_f(\rho|_{M_\alpha}\|\sigma|_{M_\alpha})\ \nearrow\  S_f(\rho\|\sigma).
$$
\end{itemize}
\end{thm}

\begin{proof}
To prove (i)--(iv), we apply expression \eqref{F-3.12} with $L=M$. Since $(1+s)d\nu_n(s)$ is
a finite positive measure supported on $[1/n,n]$ (see \eqref{F-3.10}), it is clear that the
function of $(\rho,\sigma)$ inside the bracket $[\cdots]$ in \eqref{F-3.12} is linear and
continuous in the $\sigma(M_*,M)$-topology, so (i) and (ii) are shown. Here note that joint
convexity and joint subadditivity in (ii) are equivalent due to the homogeneity property in
\eqref{F-2.8}.

Assume that $f(0^+)\le0$; then $f_n(0^+)\le f(0^+)\le0$ for all $n\in\bN$. Hence the first
assertion of (iii) is obvious by expression \eqref{F-3.12}, and the second assertion is
similar.

To prove (iv), note first that $\rho\circ\Phi,\sigma\circ\Phi\in N_*^+$ since $\Phi$ is a
normal positive linear map. For any $N$-valued step function $x(\cdot)$ on $(0,+\infty)$,
let $y(s):=\Phi(x(s))$, which is an $M$-valued step function. Since $\Phi$ is a unital
Schwarz map, one has
\begin{align*}
\sigma((1-y(s))^*(1-y(s)))&\le\sigma\circ\Phi((1-x(s))^*(1-x(s))), \\
\rho(y(s)y(s)^*)&\le\rho\circ\Phi(x(s)x(s)^*),
\end{align*}
which implies that the bracket $[\cdots]$ in \eqref{F-3.12} for $\rho\circ\Phi$,
$\sigma\circ\Phi$ and $x(\cdot)$ is dominated by $S_f(\rho\|\sigma)$. Hence inequality
\eqref{F-4.1} follows. When $N$ is a unital von Neumann subalgebra, applying \eqref{F-4.1}
to the injection $\Phi:N\hookrightarrow M$ gives \eqref{F-4.2}.

To prove (v), apply \eqref{F-3.12} with $L=\bigcup_\alpha M_\alpha$. When we restrict
$x(\cdot)$ in \eqref{F-3.12} to $M_\alpha$-valued step functions, we have the expression of
$S_f(\rho|_{M_\alpha}\|\sigma|_{M_\alpha})$. This shows that
$S_f(\rho|_{M_\alpha}\|\sigma|_{M_\alpha})$ is increasing and
$S_f(\rho|_{M_\alpha}\|\sigma|_{M_\alpha})\le S_f(\rho\|\sigma)$. (This also follows from
monotonicity in \eqref{F-4.2}.) Hence it remains to show that
$S_f(\rho\|\sigma)\le\sup_\alpha S_f(\rho|_{M_\alpha}\|\sigma|_{M_\alpha})$. For any
$c<S_f(\rho\|\sigma)$ we choose an $n\in\bN$ and an $L$-valued step function $x(\cdot)$ such
that
\begin{align*}
c&<f_n(0^+)\sigma(1)+f_n'(+\infty)\rho(1) \\
&\quad-\int_{[1/n,n]}\bigl\{\sigma((1-x(s))^*(1-x(s)))
+s^{-1}\rho(x(s)x(s)^*)\bigr\}(1+s)\,d\nu_n(s).
\end{align*}
Since $x(\cdot)$ is $M_\alpha$-valued for some $\alpha$, we have
$c<S_f(\rho|_{M_\alpha}\|\sigma|_{M_\alpha})$, implying the desired conclusion.
\end{proof}

The next corollary shows that $S_f(\rho\|\sigma)$ is strictly positive in some typical
situation.

\begin{cor}\label{C-4.2}
Let $\rho,\sigma\in M_*^+$.
\begin{itemize}
\item[(1)] The \emph{Peierls-Bogolieubov inequality} holds:
\begin{align}\label{F-4.3}
S_f(\rho\|\sigma)\ge\sigma(1)f(\rho(1)/\sigma(1)).
\end{align}
Assume that $f$ is non-linear and $\rho,\sigma\ne0$. Then equality holds in \eqref{F-4.3}
if and only if $\rho=(\rho(1)/\sigma(1))\sigma$.
\item[(2)] \emph{Strict positivity:} Assume that $f$ is non-linear with $f(1)=0$ and
$\rho(1)=\sigma(1)>0$. Then $S_f(\rho\|\sigma)\ge0$, and $S_f(\rho\|\sigma)=0$ $\iff$
$\rho=\sigma$.
\end{itemize}
\end{cor}

\begin{proof}
(1)\enspace
When $N:=\bC1$ in \eqref{F-4.2}, inequality \eqref{F-4.3} arises. If $\rho=k\sigma$ with a
constant $k>0$, then we have $\Delta_{\rho,\sigma}=k\Delta_\sigma$, $\Delta_\sigma$ being
the modular operator for $\sigma$, and hence
$d\|E_{\rho,\sigma}(t)\xi_\sigma\|^2=\sigma(1)d\delta_1(t)$, giving
$S_f(\rho\|\sigma)=f(k)\sigma(1)$. Conversely, assume that $\rho,\sigma\ne0$ and equality
holds in \eqref{F-4.3}. Further, assume that $f$ is non-linear. Since $f$ is operator convex
on $(0,+\infty)$, it is strictly convex there. For every projection $e\in M$, applying
\eqref{F-4.2} to $N:=\bC e+\bC e^\perp$ (where $e^\perp:=1-e$) gives
$$
S_f(\rho\|\sigma)\ge S_f(\rho|_N\|\sigma|_N)=\sigma(e)f(\rho(e)/\sigma(e))
+\sigma(e^\perp)f(\rho(e^\perp)/\sigma(e^\perp)).
$$
From this and \eqref{F-4.3} for $\rho|_N$ and $\sigma|_N$ one has
$$
\sigma(1)f(\rho(1)/\sigma(1))
=\sigma(e)f(\rho(e)/\sigma(e))+\sigma(e^\perp)f(\rho(e^\perp)/\sigma(e^\perp)).
$$
By Lemma \ref{L-4.3} below one has $\rho(e)=k\sigma(e)$ and $\rho(e^\perp)=k\sigma(e^\perp)$
for some $k>0$. Since $k=\rho(1)/\sigma(1)$ follows, we find that
$\rho(e)=(\rho(1)/\sigma(1))\sigma(e)$ for all projections $e\in M$, showing that
$\rho=(\rho(1)/\sigma(1))\sigma$.

(2) is immediately seen from (1).
\end{proof}

The next elementary lemma has been used in the above, whose proof is given for completeness,
since we find no suitable reference.

\begin{lemma}\label{L-4.3}
Let $f:(0,+\infty)\to\bR$ is a strictly convex function (not necessarily operator convex).
Let $a_i,b_i\in[0,+\infty)$ for $i=1,2$ be such that $a_1+a_2>0$ and $b_1+b_2>0$. If
$$
(b_1+b_2)f\biggl({a_1+a_2\over b_1+b_2}\biggr)=b_1f(a_1/b_1)+b_2f(a_2/b_2)
$$
with convention \eqref{F-2.3}, then $(a_1,a_2)=k(b_1,b_2)$ for some $k>0$.
\end{lemma}

\begin{proof}
We may consider the following four cases separately.

{\it Case $a_i,b_i>0$ for $i=1,2$}.\enspace
Since
$$
(b_1+b_2)f\biggl({b_1\over b_1+b_2}\cdot{a_1\over b_1}
+{b_2\over b_1+b_2}\cdot{a_2\over b_2}\biggr)=b_1f(a_1/b_1)+b_2f(a_2/b_2),
$$
the strict convexity of $f$ implies that $a_1/b_1=a_2/b_2$.

{\it Case $a_1=0$ and $b_1,b_2>0$}.\enspace
The assumption means that
$$
(b_1+b_2)f\biggl({a_2\over b_1+b_2}\biggr)=b_1f(0/b_1)+b_2f(a_2/b_2)
=b_1f(0^+)+b_2f(a_2/b_2),
$$
which implies that $f(0^+)<+\infty$. Hence $f$ extends to a strictly convex function on
$[0,+\infty)$, and the above equality gives $a_2/b_2=0$, which is impossible since
$a_1+a_2>0$.

{\it Case $a_1,a_2>0$ and $b_1=0$}.\enspace
This case reduces to the previous case if we replace $f$ with its transpose $\widetilde f$.

{\it Case $a_1=b_1=0$ or $a_2=b_2=0$}.\enspace
The assertion trivially holds in this case.

{\it Case $a_1=b_2=0$}.\enspace
The assumption means that
$$
b_1f(a_2/b_1)=b_1f(0/b_1)+0f(a_2/0)=f(0^+)b_1+f'(+\infty)a_2,
$$
which implies that $f(0^+)<+\infty$ and $f'(+\infty)<+\infty$. Then it is easy to find that
$f(t)<f(0^+)+f'(+\infty)t$ for all $t>0$, which contradicts the above equality.
\end{proof}

For $\sigma\in M_*^+$ and a projection $e\in M$, we write $e\sigma e$ for the restriction of
$\sigma$ to the reduced von Neumann algebra $eMe$. 

\begin{cor}\label{C-4.4}
\begin{itemize}
\item[(1)] If $e\in M$ is a projection such that $s_M(\rho),s_M(\sigma)\le e$, then
\begin{align}\label{F-4.4}
S_f(\rho\|\sigma)=S_f(e\rho e\|e\sigma e).
\end{align}
\item[(2)] If $\rho_i,\sigma_i\in M_*^+$, $i=1,2$, and
$s_M(\rho_1)\vee s_M(\sigma_1)\perp s_M(\rho_2)\vee s_M(\sigma_2)$, then
\begin{align}\label{F-4.5}
S_f(\rho_1+\rho_2\|\sigma_1+\sigma_2)=S_f(\rho_1\|\sigma_1)+S_f(\rho_2\|\sigma_2).
\end{align}
\item[(3)] If $\omega_1,\omega_2\in M_*^+$ and $S_f(\omega_1\|\omega_2)<+\infty$, then for
every $\rho,\sigma\in M_*^+$,
$$
S_f(\rho\|\sigma)=\lim_{\eps\searrow0}S_f(\rho+\eps\omega_1\|\sigma+\eps\omega_2).
$$
In particular, for every $\rho,\sigma,\omega\in M_*^+$,
\begin{align}\label{F-4.6}
S_f(\rho\|\sigma)=\lim_{\eps\searrow0}S_f(\rho+\eps\omega\|\sigma+\eps\omega).
\end{align}
\end{itemize}
\end{cor}

\begin{proof}
(1)\enspace
By monotonicity \eqref{F-4.2} we have $S_f(e\rho e\|e\sigma e)\le S_f(\rho\|\sigma)$. For any
$M$-valued step function $x(\cdot)$ on $(0,+\infty)$, let $y(s):=ex(s)e$, which is an
$eMe$-valued step function. Since
\begin{align*}
(e\sigma e)((e-y(s))^*(e-y(s)))&=\sigma(e(1-x(s))^*e(1-x(s))e)
\le\sigma((1-x(s))^*(1-s(s))), \\
(e\rho e)(y(s)y(s)^*)&=\sigma(ex(s)ex(s)^*e)\le\sigma(x(s)x(s)^*),
\end{align*}
the bracket $[\cdots]$ in \eqref{F-3.12} for $x(\cdot)$ is dominated by
$S_f(e\rho e\|e\sigma e)$. Hence equality \eqref{F-4.4} follows.

(2)\enspace
Let $e:=s_M(\rho_1)\vee s_M(\rho_2)$ and so $s_M(\rho_2)\vee s_M(\sigma_2)\le e^\perp$. Let
$\rho:=\rho_1+\rho_2$, $\sigma:=\sigma_1+\sigma_2$, and $N:=eMe\oplus e^\perp Me^\perp$.
Since $\rho|_N=e\rho_1e\oplus e^\perp\rho_2e^\perp$ and
$\sigma|_N=e\sigma_1e\oplus e^\perp\sigma_2e^\perp$, from monotonicity \eqref{F-4.2} and
Proposition \ref{P-2.3}\,(4) one has
$$
S_f(\rho\|\sigma)\ge S_f(e\rho_1e\|e\sigma_1e)
+S_f(e^\perp\rho_2e^\perp\|e^\perp\sigma_2e^\perp)
=S_f(\rho_1\|\sigma_1)+S_f(\rho_2\|\sigma_2),
$$
where we have used (1) for the last equality. On the other hand, consider the map
$$
\Phi: M\,\longrightarrow\,eMe\oplus e^\perp Me^\perp,\quad
\Phi(x):=exe+e^\perp xe^\perp,
$$
which is unital and completely positive (hence a Schwarz map). Since
$\rho=(e\rho_1e\oplus e^\perp\rho_2e^\perp)\circ\Phi$ and
$\sigma=(e\sigma_1e\oplus e^\perp\sigma_2e^\perp)\circ\Phi$, from monotonicity \eqref{F-4.1}
and Proposition \ref{P-2.3}\,(4) one has
$$
S_f(\rho\|\sigma)\le S_f(\rho_1\|\sigma_1)+S_f(\rho_2\|\sigma_2).
$$
Hence equality \eqref{F-4.5} is shown.

(3)\enspace
From joint subadditivity in Theorem \ref{T-4.1}\,(ii) and homogeneity \eqref{F-2.9} one has
$$
S_f(\rho+\eps\omega_1\|\sigma+\eps\omega_2)\le
S_f(\rho\|\sigma)+\eps S_f(\omega_1\|\omega_2)
\,\longrightarrow\,S_f(\rho\|\sigma)
$$
as $\eps\searrow0$. On the other hand, from lower semicontinuity in Theorem \ref{T-4.1}\,(i)
one has
$$
S_f(\rho\|\sigma)\le\liminf_{\eps\searrow0}S_f(\rho+\eps\omega_1\|\sigma+\eps\omega_2),
$$
showing the asserted convergence.
\end{proof}

The additivity in Corollary \ref{C-4.4}\,(2) improves that in Proposition \ref{P-2.3}\,(4);
yet we have used the latter in the above proof of the former. When $M$ is $\sigma$-finite so
that a faithful $\omega\in M_*^+$ exists, we can sometimes reduce arguments on
$S_f(\rho\|\sigma)$ to the case of faithful $\rho,\sigma\in M_*^+$ by using the convergence
property in \eqref{F-4.6}.

The next continuity property is not included in the martingale convergence in Theorem
\ref{T-4.1}, since $eMe$ is not a unital von Neumann subalgebra of $M$.

\begin{thm}\label{T-4.5}
Let $\{e_\alpha\}$ be an increasing net of projections in $M$ such that $e_\alpha\nearrow1$.
Then for every $\rho,\sigma\in M_*^+$,
$$
\lim_\alpha S_f(e_\alpha\rho e_\alpha\|e_\alpha\sigma e_\alpha)=S_f(\rho\|\sigma).
$$
\end{thm}

\begin{proof}
By replacing $f$ with $f(t)-(a+bt)$ and noting \eqref{F-2.7}, we may assume that $f(t)\ge0$
for all $t\in(0,+\infty)$. Let $M_\alpha:=e_\alpha Me_\alpha+\bC(1- e_\alpha)$; then
$\{M_\alpha\}$ is an increasing net of von Neumann subalgebras with
$\bigl(\bigcup_\alpha M_\alpha\bigr)''=M$. Hence the martingale convergence in Theorem
\ref{T-4.1} and Corollary \ref{C-4.4}\,(2) imply that
\begin{align}\label{F-4.7}
S_f(\rho\|\sigma)
=\lim_\alpha\biggl[S_f(e_\alpha\rho e_\alpha\|e_\alpha\sigma e_\alpha)
+\sigma(1-e_\alpha)f\biggl({\rho(1-e_\alpha)\over\sigma(1-e_\alpha)}\biggr)\biggr],
\end{align}
increasingly in $\alpha$. First, assume that $S_f(\rho\|\sigma)=+\infty$ and prove that
\begin{align}\label{F-4.8}
\lim_\alpha S_f(e_\alpha\rho e_\alpha\|e_\alpha\sigma e_\alpha)=+\infty.
\end{align}
If $\limsup_\alpha\sigma(1-e_\alpha)f(\rho(1-e_\alpha)/\sigma(1-e_\alpha))<+\infty$, then
\eqref{F-4.8} clearly follows from \eqref{F-4.7}. Assume that
$\limsup_\alpha\sigma(1-e_\alpha)f(\rho(1-e_\alpha)/\sigma(1-e_\alpha))=+\infty$. Then for
any $K>0$ choose an $\alpha_0$ such that
$\sigma(1-e_{\alpha_0})f(\rho(1-e_{\alpha_0})/\sigma(1-e_{\alpha_0}))>K$. Since
$$
\rho(e_\alpha-e_{\alpha_0})\nearrow\rho(1-e_{\alpha_0}),\quad
\sigma(e_\alpha-e_{\alpha_0})\nearrow\sigma(1-{\alpha_0})\qquad
\mbox{as $\alpha_0\le\alpha\to``\infty"$},
$$
we easily see that
$$
\sigma(e_\alpha-e_{\alpha_0})f\biggl({\rho(e_\alpha-e_{\alpha_0})\over
\sigma(e_\alpha-e_{\alpha_0})}\biggr)
\ \longrightarrow\ \sigma(1-e_{\alpha_0})f\biggl({\rho(1-e_{\alpha_0})\over
\sigma(1-e_{\alpha_0})}\biggr)>K.
$$
By monotonicity of $S_f$ and the assumption $f\ge0$ we have for $\alpha\ge\alpha_0$
\begin{align*}
S_f(e_\alpha\rho e_\alpha\|e_\alpha\sigma e_\alpha)
&\ge\sigma(e_\alpha-e_{\alpha_0})f\biggl({\rho(e_\alpha-e_{\alpha_0})\over
\sigma(e_\alpha-e_{\alpha_0})}\biggr)
+\sigma(e_{\alpha_0})f\biggl({\rho(e_{\alpha_0})\over\sigma(e_{\alpha_0})}\biggr) \\
&\ge\sigma(e_\alpha-e_{\alpha_0})f\biggl({\rho(e_\alpha-e_{\alpha_0})\over
\sigma(e_\alpha-e_{\alpha_0})}\biggr),
\end{align*}
which is $>K$ for all sufficiently large $\alpha\ge\alpha_0$. Hence \eqref{F-4.8} follows.

Next, assume that $S_f(\rho\|\sigma)<+\infty$, and prove that
$\lim_\alpha S_f(e_\alpha\rho e_\alpha\|e_\alpha\sigma e_\alpha)=S_f(\rho\|\sigma)$. To do
this, by \eqref{F-4.7} we may prove that
$\lim_\alpha\sigma(1-e_\alpha)f(\rho(1-e_\alpha)/\sigma(1-e_\alpha))=0$. Assume on the
contrary that $\limsup_\alpha\sigma(1-e_\alpha)f(\rho(1-e_\alpha)/\sigma(1-e_\alpha))>\eps>0$
for some $\eps>0$ (here recall that $f\ge0$ was assumed). Choose an $\alpha_1$ such that
$\sigma(1-e_{\alpha_1})f(\rho(1-e_{\alpha_1})/\sigma(1-e_{\alpha_1}))>\eps$.
Then we can choose a $\beta_1>\alpha_1$ such that
$$
\sigma(e_{\beta_1}-e_{\alpha_1})f\biggl({\rho(e_{\beta_1}-e_{\alpha_1})\over
\sigma(e_{\beta_1}-e_{\alpha_1})}\biggr)>\eps.
$$
Next, choose an $\alpha_1>\beta_2$ such that
$\sigma(1-e_{\alpha_2})f(\rho(1-e_{\alpha_2})/\sigma(1-e_{\alpha_2}))>\eps$, and
a $\beta_2>\alpha_2$ such that
$$
\sigma(e_{\beta_2}-e_{\alpha_2})f\biggl({\rho(e_{\beta_2}-e_{\alpha_2})\over
\sigma(e_{\beta_2}-e_{\alpha_2})}\biggr)>\eps.
$$
Repeating the above argument we have $\alpha_1<\beta_1<\alpha_2<\beta_2<\cdots$ in such a
way that
$$
\sigma(e_{\beta_k}-e_{\alpha_k})f\biggl({\rho(e_{\beta_k}-e_{\alpha_k})\over
\sigma(e_{\beta_k}-e_{\alpha_k})}\biggr)>\eps
$$
for all $k\in\bN$. Let $e_{\alpha_k}\nearrow e_\infty$ and $e_0:=1-e_\infty$, and consider a
unital abelian von Neumann subalgebra of $M$
$$
\bigoplus_{k=1}^\infty\bC(e_{\beta_k}-e_{\alpha_k})\oplus
\bigoplus_{k=1}^\infty\bC(e_{\alpha_k}-e_{\beta_{k-1}})\oplus\bC e_0,
$$
where $e_{\beta_0}:=0$. By monotonicity of $S_f$ and Example \ref{E-2.5} together with
$f\ge0$, we have
$$
S_f(\rho\|\sigma)\ge\sum_{k=1}^\infty
\sigma(e_{\beta_k}-e_{\alpha_k})f\biggl({\rho(e_{\beta_k}-e_{\alpha_k})\over
\sigma(e_{\beta_k}-e_{\alpha_k})}\biggr)
+\sigma(e_0)f\biggl({\rho(e_0)\over\sigma(e_0)}\biggr)
=+\infty,
$$
which contradicts the assumption $S_f(\rho\|\sigma)<+\infty$.
\end{proof}

When $f\ge0$ in Theorem \ref{T-4.5}, from the monotonicity of $S_f$ we see that
$S_f(e_\alpha\rho e_\alpha\|e_\alpha\sigma e_\alpha)$ is increasing as $e_\alpha\nearrow1$.
But this is not the case unless $f\ge0$.

\begin{remark}\rm
When $M=B(\cH)$ with $\dim\cH=\infty$, according to Theorem \ref{T-4.5}, one can define the
relative entropy $D(\rho\|\sigma)$ for trace-class operators $\rho,\sigma\ge0$ as
\begin{align}\label{F-4.9}
D(\rho\|\sigma)=\lim_\alpha D(E_\alpha\rho E_\alpha\|E_\alpha\sigma E_\alpha),
\end{align}
where $\{E_\alpha\}$ is an increasing net of finite rank projections with
$E_\alpha\nearrow I$. But it seems that there is no simpler proof other than that of Theorem
\ref{T-4.5} for the existence of the limit in \eqref{F-4.9} and its independence of the
choice of $\{E_\alpha\}$.
\end{remark}

\section{R\'enyi divergences}

We define the notion of R\'enyi divergences $D_\alpha(\rho\|\sigma)$ for $\alpha\ge0$ in the
general von Neumann algebra setting.

\begin{definition}\label{D-5.1}\rm
Let $\rho,\sigma\in M_*^+$. Since $\xi_\sigma\in\cD(\Delta_{\rho,\sigma}^{1/2})$, the domain
of $\Delta_{\rho,\sigma}^{1/2}$, note that $\xi_\sigma\in\cD(\Delta_{\rho,\sigma}^{\alpha/2})$
for any $\alpha\in[0,1]$. We define the quantities $Q_\alpha(\rho\|\sigma)$ for $\alpha\ge0$
as follows: When $0\le\alpha<1$,
\begin{align}\label{F-5.1}
Q_\alpha(\rho\|\sigma):=\|\Delta_{\rho,\sigma}^{\alpha/2}\xi_\sigma\|^2
\ \in[0,+\infty),
\end{align}
and when $\alpha>1$,
\begin{align}\label{F-5.2}
Q_\alpha(\rho\|\sigma):=\begin{cases}
\|\Delta_{\rho,\sigma}^{\alpha/2}\xi_\sigma\|^2 &
\text{if $s_M(\rho)\le s_M(\sigma)$ and
$\xi_\sigma\in\cD(\Delta_{\rho,\sigma}^{\alpha/2})$}, \\
+\infty & \text{otherwise}.
\end{cases}
\end{align}
Moreover, when $\alpha=1$, define $Q_1(\rho\|\sigma):=\rho(1)$. Then for every
$\rho,\sigma\in M_*^+$ with $\rho\ne0$ and for each $\alpha\in[0,+\infty)\setminus\{1\}$,
the \emph{$\alpha$-R\'enyi divergence} $D_\alpha(\rho\|\sigma)$ is defined by
\begin{align}\label{F-5.3}
D_\alpha(\rho\|\sigma):={1\over\alpha-1}\log{Q_\alpha(\rho\|\sigma)\over\rho(1)}.
\end{align}
\end{definition}

In particular, note that $Q_0(\alpha\|\sigma)=\sigma(s_M(\rho))$ and
$D_0(\rho\|\sigma)=-\log\bigl[\sigma(s_M(\rho))/\rho(1)\bigr]$. The next lemma shows that
$Q_\alpha(\rho\|\sigma)$ is essentially a standard $f$-divergence and so
$D_\alpha(\rho\|\sigma)$ is a variant of standard $f$-divergences.

\begin{lemma}\label{L-5.2}
Define convex functions $f_\alpha$ on $[0,+\infty)$ by
$$
f_\alpha(t):=\begin{cases}t^\alpha & \text{if $\alpha\ge1$}, \\
-t^\alpha & \text{if $0<\alpha<1$}.
\end{cases}
$$
Then for every $\rho,\sigma\in M_*^+$, $Q_\alpha(\rho\|\sigma)$ is given as
\begin{align}\label{F-5.4}
Q_\alpha(\rho\|\sigma)=\begin{cases}
S_{f_\alpha}(\rho\|\sigma) & \text{if $\alpha\ge1$}, \\
-S_{f_\alpha}(\rho\|\sigma) & \text{if $0<\alpha<1$}.
\end{cases}
\end{align}
Moreover,
\begin{align}\label{F-5.5}
Q_\alpha(\rho\|\sigma)=\begin{cases}
\int_{(0,+\infty)}t^\alpha\,d\|E_{\rho,\sigma}(t)\xi_\sigma\|^2
& \text{if $0\le\alpha<1$ or $s_M(\rho)\le s_M(\sigma)$}, \\
+\infty & \text{if $\alpha>1$ and $s_M(\rho)\not\le s_M(\sigma)$}.
\end{cases}
\end{align}
\end{lemma}

\begin{proof}
When $0<\alpha<1$, since $f_\alpha(0)=f_\alpha'(+\infty)=0$, we have by \eqref{F-5.1}
$$
S_{f_\alpha}(\rho\|\sigma)
=\int_{(0,+\infty)}(-t^\alpha)\,d\|E_{\rho,\sigma}(t)\xi_\sigma\|^2
=-Q_\alpha(\rho\|\sigma).
$$
When $\alpha=1$, \eqref{F-2.7} gives $S_{f_1}(\rho\|\sigma)=\rho(1)=Q_1(\rho\|\sigma)$.
When $\alpha>1$, since $f_\alpha(0)=0$ and $f_\alpha'(+\infty)=+\infty$,
$$
S_{f_\alpha}(\rho\|\sigma)
=\int_{(0,+\infty)}t^\alpha\,d\|E_{\rho,\sigma}(t)\xi_\sigma\|^2
+(+\infty)\rho(1-s_M(\sigma)).
$$
Note that $\rho(1-s_M(\sigma))=0$ $\iff$ $s_M(\rho)\le s_M(\sigma)$, and
$\int_{(0,+\infty)}t^\alpha\,d\|E_{\rho,\sigma}(t)\xi_\sigma\|^2<+\infty$ $\iff$
$\xi_\sigma\in\cD(\Delta_{\rho,\sigma}^{\alpha/2})$. Hence by \eqref{F-5.2} we see that
$S_{f_\alpha}(\rho\|\sigma)=Q_\alpha(\rho\|\sigma)$. Moreover, \eqref{F-5.5} immediately
follows from the above argument, where the case $\alpha=0$ is obvious.
\end{proof}

Some properties of $Q_\alpha$ and $D_\alpha$ are found in, e.g., \cite{Li,MH,Hi0,Mo} though
mostly in the finite-dimensional situation. More comprehensive summary of them are given in
the next proposition, mainly based on Theorem \ref{T-4.1}. Although (3) and (4) when
$\alpha\in[0,2]$ have been shown in \cite{BST}, we give their proofs as well for convenience
of the reader.

\begin{prop}\label{P-5.3}
Let $\rho,\sigma\in M_*^+$ with $\rho\ne0$.
\begin{itemize}
\item[(1)] If $s_M(\rho)\perp s_M(\sigma)$, then $D_\alpha(\rho\|\sigma)=+\infty$ for all
$\alpha\in[0,+\infty)\setminus\{1\}$.
\item[(2)] If $s_M(\rho)\not\perp s_M(\sigma)$, then $Q_\alpha(\rho\|\sigma)>0$ for all
$\alpha\ge0$ and the function $\alpha\in[0,+\infty)\mapsto\log Q_\alpha(\rho\|\sigma)$ is
convex.
\item[(3)] The limit $D_1(\rho\|\sigma):=\lim_{\alpha\nearrow1}D_\alpha(\rho\|\sigma)$ exists
and
\begin{align}\label{F-5.6}
D_1(\rho\|\sigma)={D(\rho\|\sigma)\over\rho(1)},
\end{align}
where $D(\rho\|\sigma)$ is the relative entropy. Moreover, if $D_\alpha(\rho\|\sigma)<+\infty$
for some $\alpha>1$, then $\lim_{\alpha\searrow1}D_\alpha(\rho\|\sigma)=D_1(\rho\|\sigma)$.
\item[(4)] The function $\alpha\in[0,+\infty)\mapsto D_\alpha(\rho\|\sigma)$ is monotone
increasing.
\item[(5)] Assume that $0<\alpha<1$. We have
\begin{align}\label{F-5.7}
Q_\alpha(\rho\|\sigma)=Q_{1-\alpha}(\rho\|\sigma),
\end{align}
and whenever $\rho,\sigma\ne0$,
\begin{align}\label{F-5.8}
{1\over\alpha}\,D_\alpha(\rho\|\sigma)={1\over1-\alpha}\,D_{1-\alpha}(\sigma\|\rho)
+{1\over\alpha(1-\alpha)}\log{\rho(1)\over\sigma(1)}.
\end{align}
Hence, if $\rho(1)=\sigma(1)$, then
$\lim_{\alpha\searrow0}{1\over\alpha}\,D_\alpha(\rho\|\sigma)=D_1(\sigma\|\rho)$.
\item[(6)] \emph{Joint lower semicontinuity:} For every $\alpha\in[0,2]$, the map
$(\rho,\sigma)\in(M_*^+\setminus\{0\})\times M_*^+\mapsto D_\alpha(\rho\|\sigma)\in(-\infty,+\infty]$ is
jointly lower semicontinuous in the $\sigma(M_*,M)$-topology.
\item[(7)] The map $(\rho,\sigma)\in M_*^+\times M_*^+\mapsto Q_\alpha(\rho\|\sigma)$ is
jointly concave and jointly superadditive for $0\le\alpha\le1$, and jointly convex and
jointly subadditive for $1\le\alpha\le2$. Hence, when $0\le\alpha\le1$,
$D_\alpha(\rho\|\sigma)$ is jointly convex on
$\{(\rho,\sigma)\in M_*^+\times M_*^+:\rho(1)=c\}$ for any fixed $c>0$.
\item[(8)] Let $\rho_i,\sigma_i\in M_*^+$ for $i=1,2$. If $0\le\alpha<1$, $\rho_1\le\rho_2$
and $\sigma_1\le\sigma_2$, then $Q_\alpha(\rho_1\|\sigma_1)\le Q_\alpha(\rho_2\|\sigma_2)$.
If $1\le\alpha\le2$ and $\sigma_1\le\sigma_2$, then
$Q_\alpha(\rho\|\sigma_1)\ge Q_\alpha(\rho\|\sigma_2)$. If $\sigma_1\le\sigma_2$, then
$D_\alpha(\rho\|\sigma_1)\ge D_\alpha(\rho\|\sigma_2)$ for all $\alpha\in[0,2]$.
\item[(9)] \emph{Monotonicity:} For each $\alpha\in[0,2]$, $D_\alpha(\rho\|\sigma)$ is
monotone under unital normal Schwarz maps, i.e.,
\begin{align}\label{F-5.9}
D_\alpha(\rho\circ\Phi\|\sigma\circ\Phi)\le D_\alpha(\rho\|\sigma)
\end{align}
for any unital normal Schwarz map $\Phi:N\to M$ as in Theorem \ref{T-4.1}\,(iv).
\item[(10)] \emph{Strict positivity:} Let $\alpha\in(0,+\infty)$ and $\rho,\sigma\ne0$. The
inequality
\begin{align}\label{F-5.10}
D_\alpha(\rho\|\sigma)\ge\log{\rho(1)\over\sigma(1)}
\end{align}
holds, and equality holds in \eqref{F-5.10} if and only if $\rho=(\rho(1)/\sigma(1))\sigma$.
If $\rho(1)=\sigma(1)$, then $D_\alpha(\rho\|\sigma)\ge0$, and $D_\alpha(\rho\|\sigma)=0$
$\iff$ $\rho=\sigma$.
\end{itemize}
\end{prop}

\begin{proof}
Write $F(\alpha):=\int_{(0,+\infty)}t^\alpha\,d\mu(t)$ for $\alpha>0$, where
$d\mu(t):=d\|E_{\rho,\sigma}(t)\xi_\sigma\|^2$ for $t\in(0,+\infty)$. Note that
$F(1)=\rho(s_M(\sigma))$. By \eqref{F-5.5} we note that
\begin{itemize}
\item[(A)] if $s_M(\rho)\le s_M(\sigma)$ then $Q_\alpha(\rho\|\sigma)=F(\alpha)$ for all
$\alpha\ge0$,
\item[(B)] if $s_M(\rho)\not\le s_M(\sigma)$ then
$$
Q_\alpha(\rho\|\sigma)=\begin{cases}
F(\alpha) & \text{for $0\le\alpha<1$}, \\
\rho(1)>F(1) &\text{for $\alpha=1$}, \\
+\infty & \text{for $\alpha>1$}.
\end{cases}
$$
\end{itemize}

(1)\enspace
If $s_M(\rho)\perp s_M(\sigma)$, i.e., $F(1)=\rho(s_M(\sigma))=0$, then we have $\mu=0$ so
that $F(\alpha)=0$ for all $\alpha\in[0,+\infty)$. Hence the conclusion follows from (B).

(2)\enspace
If $s_M(\rho)\not\perp s_M(\sigma)$, then we have $\mu\ne0$ so that $F(\alpha)>0$ for all
$\alpha\in[0,+\infty)$. Now by (A) and (B) we may show that $\log F(\alpha)$ is convex on
$[0,+\infty)$. Let $\alpha_1,\alpha_2\in(0,+\infty)$ and $0<\lambda<1$. H\"older's inequality
implies that
$$
\int t^{(1-\lambda)\alpha_1+\lambda t_2}\,d\mu(t)
\le\biggl[\int t^{\alpha_1}\,d\mu(t)\biggr]^{1-\lambda}
\biggl[\int t^{\alpha_2}\,d\mu(t)\biggr]^\lambda,
$$
which shows the convexity of $\log F(\alpha)$.

(3)\enspace
First, assume that $s_M(\rho)\not\le s_M(\sigma)$. As $0<\alpha\nearrow1$, since
$t^\alpha\nearrow t$ for $t\ge1$, by the monotone convergence theorem, we have
$F(\alpha)\to F(1)<\rho(1)$. Hence
$$
D_\alpha(\rho\|\sigma)={\log F(\alpha)-\log\rho(1)\over\alpha-1}
\ \longrightarrow\ +\infty={D(\rho\|\sigma)\over\rho(1)}.
$$
Second, assume that $s_M(\rho)\le s_M(\sigma)$, i.e., $F(1)=\rho(1)$. For any $t\in(0,+\infty)$,
since $\alpha\in(0,+\infty)\mapsto t^\alpha$ is convex, we see that as $0<\alpha\nearrow1$,
$$
t-1\le{t^\alpha-t\over\alpha-1}\ \nearrow\ t\log t,\qquad t\in(0,+\infty),
$$
so that the monotone convergence theorem gives
\begin{align}\label{F-5.11}
{F(\alpha)-F(1)\over\alpha-1}=\int{t^\alpha-t\over\alpha-1}\,d\mu(t)
\ \nearrow\ \int t\log t\,d\mu(t)=S_{t\log t}(\rho\|\sigma)=D(\rho\|\sigma).
\end{align}
This implies that as $0<\alpha\nearrow1$,
$$
F(\alpha)=\rho(1)+D(\rho\|\sigma)(\alpha-1)+o(1-\alpha)
$$
so that
$$
\log{F(\alpha)\over\rho(1)}={D(\rho\|\sigma)\over\rho(1)}(\alpha-1)+o(1-\alpha).
$$
Therefore, $D_\alpha(\rho\|\sigma)\to D(\rho\|\sigma)/\rho(1)$.

Next, assume that $D_{\alpha_0}(\rho\|\sigma)<+\infty$, i.e., $s_M(\rho)\le s_M(\sigma)$ and
$\int t^{\alpha_0}\,d\mu(t)<+\infty$ for some $\alpha_0>1$. As $\alpha_0\ge\alpha\searrow1$,
since
$$
{t^{\alpha_0}-1\over\alpha_0-1}\ge{t^\alpha-t\over\alpha-1}\ \searrow\ t\log t,
\qquad t\in(0,+\infty),
$$
the Lebesgue convergence theorem gives, as in \eqref{F-5.11},
$$
{F(\alpha)-F(1)\over\alpha-1}\ \searrow\ D(\rho\|\sigma),
$$
and hence the latter assertion is shown similarly to the above.

(4)\enspace
When $s_M(\rho)\perp s_M(\sigma)$, this is obvious from (1). Otherwise, this immediately
follows from convexity of $\alpha\mapsto\log Q_\alpha(\rho\|\sigma)$ in (2) (and from
definition of $D_\alpha$, $Q_1$ and $D_1$).

(5)\enspace
Let $0<\alpha<1$. Since $\widetilde f_\alpha=f_{1-\alpha}$, Proposition \ref{P-2.4} with
\eqref{F-5.4} gives $Q_\alpha(\rho\|\sigma)=Q_{1-\alpha}(\sigma\|\rho)$ and hence for
$\sigma\ne0$,
$$
D_\alpha(\rho\|\sigma)={1\over\alpha-1}\log{Q_{1-\alpha}(\sigma\|\rho)\over\rho(1)}
={\alpha\over1-\alpha}\,D_{1-\alpha}(\sigma\|\rho)
+{1\over1-\alpha}\log{\rho(1)\over\sigma(1)},
$$
implying \eqref{F-5.8}. From this and (3) the second assertion follows.

(6)\enspace
One can consider $\log$ as an continuous increasing function from $[0,+\infty]$ to
$[-\infty,+\infty]$. We see from (A) and (B) above that for every $\rho,\sigma\in M_*^+$,
$Q_\alpha(\rho\|\sigma)$ is in $[0,+\infty)$ for $0\le\alpha<1$ and in $(0,+\infty]$ for
$\alpha>1$. Hence $D_\alpha(\rho\|\sigma)\in(-\infty,+\infty]$ for any
$\alpha\in[0,+\infty)\setminus\{1\}$. Since $f_\alpha$ is operator convex on $(0,+\infty)$ if
$0\le\alpha\le2$, by \eqref{F-5.4} and Theorem \ref{T-4.1}\,(i) the map
$(\rho,\alpha)\in M_*^+\times M_*^+\mapsto Q_\alpha(\rho\|\sigma)$ is upper semicontinuous
for $0\le\alpha<1$ and lower semicontinuous for $1<\alpha\le2$ in the $\sigma(M_*,M)$-topology.
Hence $(\rho,\sigma)\in(M_*^+\setminus\{0\})\times M_*^+\mapsto D_\alpha(\rho\|\sigma)$ is
lower semicontinuous for any $\alpha\in[0,2]\setminus\{1\}$. For $\alpha=1$, the result reduces
to the case of the relative entropy due to \eqref{F-5.6}. 

(7)\enspace
The first part is a consequence of Theorem \ref{T-4.1}\,(ii) in view of \eqref{F-5.4}. Then
the second is clear since a non-negative concave function is log-concave.

(8)\enspace
The results for $Q_\alpha$ follow from Theorem \ref{T-4.1}\,(iii) and \eqref{F-5.4}. This gives
the assertion on $D_\alpha$ for $\alpha\in[0,2]\setminus\{1\}$. The case of $D_1$ follows from
(3) or it is a well-known fact of $D$.

(9) follows from Theorem \ref{T-4.1}\,(iv) in view of \eqref{F-5.4} and \eqref{F-5.6} for
$\alpha=1$.

(10)\enspace
When $\alpha\in[0,2]\setminus\{1\}$, inequality \eqref{F-5.10} is a special case of
\eqref{F-5.9} for $N=\bC1$, since for scalars $\rho(1)$ and $\sigma(1)$,
$$
D_\alpha(\rho(1)\|\sigma(1))={1\over\alpha-1}
\log{\rho(1)^\alpha\sigma(1)^{1-\alpha}\over\rho(1)}=\log{\rho(1)\over\sigma(1)}.
$$
By (4) the inequality holds for $\alpha>2$ as well. If
$\rho=k\sigma$ with $k=\rho(1)/\sigma(1)$, then
$\Delta_{\rho,\sigma}\xi_\sigma=k\Delta_\sigma\xi_\sigma=k\xi_\sigma$ and hence
$Q_\alpha(\rho\|\sigma)=k^\alpha\sigma(1)$. Therefore,
$D_\alpha(\rho\|\sigma)=\log(\rho(1)/\sigma(1))$ for all $\alpha>0$ (including $\alpha=1$).
Conversely, if equality holds for some $\alpha>0$, then by (4) the same holds for some
$\alpha\in(0,1)$. This means that equality \eqref{F-4.3} holds for $f=f_\alpha$, so that
$\rho=(\rho(1)/\sigma(1))\sigma$ follows from Corollary \ref{F-4.2}\,(1). Finally, the
second part of (10) is clear from the first.
\end{proof}

\begin{remark}\label{R-5.4}\rm
(1)\enspace
In Proposition \ref{P-5.3}\,(3), the assumption that $D_\alpha(\rho\|\sigma)<+\infty$ for some
$\alpha>1$ is essential to have
$\lim_{\alpha\searrow1}D_\alpha(\rho\|\sigma)=D_1(\rho\|\sigma)$. Indeed, it is not difficult
to find commuting density operators $\rho=\sum_{i=1}^\infty a_i|e_i\>\<e_i|$ and
$\sigma=\sum_{i=1}^\infty b_i|e_i\>\<e_i|$ on $\cH$ such that
$$
D(\rho\|\sigma)=\sum_{i=1}^\infty a_i\log{a_i\over b_i}<+\infty,\qquad
D_\alpha(\rho\|\sigma)=\sum_{i=1}^\infty a_i^\alpha b_i^{1-\alpha}=+\infty\quad
\mbox{for all $\alpha>1$}.
$$

(2)\enspace
The convexity of $Q_\alpha$ for $1\le\alpha\le2$ in Proposition \ref{P-5.3}\,(7) cannot extend
to $\alpha>2$ even in the finite-dimensional case and in separate arguments. This implies
that the monotonicity property of $D_\alpha$ in (9) fails to hold for $\alpha>2$, because the
monotonicity of $Q_\alpha$ under unital completely positive maps yields its joint convexity.
Also, the monotone decreasing of $\sigma\mapsto Q_\alpha(\rho\|\sigma)$ for $1\le\alpha\le2$
in (8) cannot extend to $\alpha>2$. But it seems possible that the joint lower semicontinuity
of $D_\alpha$ as in (6) (or in the norm topology) is true for $\alpha>2$ as well (this is easily
verified in the finite-dimensional case).

(3)\enspace
In Proposition \ref{P-5.3}\,(7), due to division by $\rho(1)$ in definition \eqref{F-5.3},
the map $\rho\mapsto D_\alpha(\rho\|\sigma)$ with $\sigma\in M_*^+$ fixed cannot be convex on
the whole $M_*^+$. However, in the finite-dimensional case it was shown \cite[Theorem II.1]{MH}
that $\sigma\mapsto D_\alpha(\rho\|\sigma)$ with $\rho\in M_*^+\setminus\{0\}$ fixed is convex
on $M_*^+$ for any $\alpha\in[0,2]$. It is natural to expect that this extends to the general
von Neumann algebra case.
\end{remark}

We end the main body of the paper with a remark on relations of $D_\alpha(\rho\|\sigma)$ with
other R\'enyi type divergences from recent papers \cite{BST,Je1}.

\begin{remark}\label{R-5.5}\rm
For every $\rho,\sigma\in M_*^+$, in view of Proposition \ref{P-5.3}\,(4) one can define
$$
D_\infty(\rho\|\sigma):=\lim_{\alpha\to+\infty}D_\alpha(\rho\|\sigma).
$$
The \emph{max-relative entropy} introduced in \cite{Da} is
$$
D_{\max}(\rho\|\sigma):=\log\inf\{t>0:\rho\le t\sigma\},
$$
where $\inf\emptyset=+\infty$ as usual. The \emph{sandwiched R\'enyi divergence}
$\widetilde D_\alpha(\rho\|\sigma)$ \cite{MDST,WWY} has recently been extended to the von
Neumann algebra setting by Berta, Scholz and Tomamichel \cite{BST} and
Jen\v cov\'a \cite{Je1,Je2}. From \cite{BST,Je1} we remark that for every
$\rho,\sigma\in M_*^+$,
\begin{itemize}
\item[(a)] $\widetilde D_\alpha(\rho\|\sigma)\le D_\alpha(\rho\|\sigma)$ for $\alpha>1$,
\item[(b)] $\lim_{\alpha\to+\infty}\widetilde D_\alpha(\rho\|\sigma)=D_{\max}(\rho\|\sigma)$,
\item[(c)] $D_2(\rho\|\sigma)\le D_{\max}(\rho\|\sigma)\le D_\infty(\rho\|\sigma)$.
\end{itemize}
\end{remark}

\section{Closing remarks}

In this paper we present a systematic treatment of standard $f$-divergences in the setting of
general von Neumann algebras and general operator convex functions $f$ on $(0,+\infty)$. The
main theorem is the variational expression of an arbitrary standard $f$-divergence
$S_f(\rho\|\sigma)$. We also present a comprehensive account on the quantum R\'enyi divergence
in von Neumann algebras on the basis of theory of standard $f$-divergences. There are some
other important quantum divergences; in particular, the maximal $f$-divergence (discussed
in \cite{HiMo} in the finite-dimensional case) is worth studying. The most significant problem
related to the standard $f$-divergence and other quantum divergences is the reversibility via
them, as explained in the Introduction. These should be our forthcoming research topics.

\subsection*{Acknowledgements}
This was supported in part by Grant-in-Aid for Scientific Research (C)17K05266. The author is
grateful to Mil\'an Mosonyi for suggestion on Theorem \ref{T-4.5} and Anna Jen\v cov\'a for
discussion about Remark \ref{R-A.7} below. 

\appendix

\section{R\'enyi divergences in terms of Haagerup's $L^p$-spaces}

\subsection{Haagerup's $L^p$-spaces and Connes' Radon-Nikodym cocycles}
We first give, for the convenience of the reader, a brief survey on the Haagerup $L^p$-spaces
(see \cite{Te} for details). Let us take a faithful normal semifinite weight $\ffi_0$ on $M$
and denote by $N$ the crossed product $M\rtimes_{\sigma^{\ffi_0}}\bR$ of $M$ by the modular
automorphism group $\sigma_t^{\ffi_0}=\Delta_{\ffi_0}^{it}\cdot\Delta_{\ffi_0}^{-it}$,
$t\in\bR$. Let $\theta_s$, $s\in\bR$, be the dual action of $N$ so that
$\tau\circ\theta_s=e^{-s}\tau$, $s\in\bR$, where $\tau$ is the canonical trace on $N$. Let
$\widetilde N$ denote the space of \emph{$\tau$-measurable operators} \cite{Ne,Te} affiliated
with $N$. For $0<p\le\infty$ \emph{Haagerup's $L^p$-space} $L^p(M)$ \cite{Ha2} is defined by
$$
L^p(M)=\{x\in\widetilde N: \theta_s(x)=e^{-s/p}x,\ s\in\bR\}.
$$
In particular, $L^\infty(M)=M$. Let $L^p(M)_+=L^p(M)\cap\widetilde N_+$ where $\widetilde N_+$
is the positive part of $\widetilde N$. Then $M_*$ is canonically order-isomorphic to
$L^1(M)$ by a linear bijection $\psi\in M_*\mapsto h_\psi\in L^1(M)$, so that the positive
linear functional $\tr$ on $L^1(M)$ is defined by $\tr(h_\psi)=\psi(1)$, $\psi\in M_*$.

For $0<p<\infty$ the $L^p$-(quasi-)norm $\|x\|_p$ of $x\in L^p(M)$ is given by
$\|x\|_p=\tr(|x|^p)^{1/p}$. Also $\|\cdot\|_\infty$ denotes the operator norm $\|\cdot\|$ on
$M$. When $1\le p<\infty$, $L^p(M)$ is a Banach space with the norm $\|\cdot\|_p$ and whose
dual Banach space is $L^q(M)$ where $1/p+1/q=1$ by the duality
$$
(x,y)\in L^p(M)\times L^q(M)\ \longmapsto\ \tr(xy)\ (=\tr(yx)).
$$
In particular, $L^2(M)$ is a Hilbert space with the inner product $\<x,y\>=\tr(x^*y)$
($=\tr(yx^*)$). Then $(M,L^2(M),J=\,^*,L^2(M)_+)$ becomes a standard form of $M$ where $M$
is represented on $L^2(M)$ by the left multiplication. By the uniqueness (up to unitary
equivalence) of a standard form we can always choose $(M,L^2(M),\,^*\,,L^2(M)_+)$ as a
standard form of $M$. This standard form has the advantage of enjoying the Haagerup
$L^p$-space technique. Each $\omega\in M_*^+$ is represented as
$$
\omega(x)=\tr(xh_\omega)=\<h_\omega^{1/2},xh_\omega^{1/2}\>,\qquad x\in M,
$$
with the vector representative $h_\omega^{1/2}\in L^2(M)_+$. Note that $s_M(\omega)$ and
$s_{M'}(\omega)$ are the left and right multiplications, respectively, on $L^2(M)$ of the
support projection $s(\omega)\in M$.

Note that $L^p(M)$ is independent (up to isomorphism) of the choice of $\ffi_0$ and that when
$M$ is semifinite, $L^p(M)$ coincides with the $L^p$-space in the sense of \cite{Di,Se}.

Next let us briefly recall the definition of Connes' Radon-Nikodym cocycles
(see \cite[\S3]{St}, \cite[\S VIII.3]{Ta}). For each $\ffi,\omega\in M_*^+$ the balanced
functional $\theta=\theta(\ffi,\omega)$ on $\bM_2(M)=M\otimes\bM_2(\bC)$ is given by
$$
\theta\Biggl(\sum_{i,j=1}^2x_{ij}\otimes e_{ij}\Biggr)
=\ffi(x_{11})+\omega(x_{22}),\qquad x_{ij}\in M,
$$
where $e_{ij}$ ($i,j=1,2$) are the matrix units of the $2\times2$ matrix algebra $\bM_2(\bC)$.
Note that the support projection of $\theta$ is
$s(\theta)=s(\ffi)\otimes e_{11}+s(\omega)\otimes e_{22}$ and
$s(\omega)s(\ffi)\otimes e_{21}\in s(\theta)\bM_2(M)s(\theta)$. Then
\emph{Connes' Radon-Nikodym cocycle} $[D\omega:D\ffi]_t$ ($\in s(\omega)Ms(\ffi)$) is
defined by
$$
\sigma_t^\theta(s(\omega)s(\ffi)\otimes e_{21})
=[D\omega:D\ffi]_t\otimes e_{21},\qquad t\in\bR,
$$
where $\sigma^\theta_t$ is the modular automorphism group defined on
$s(\theta)\bM_2(M)s(\theta)$. When $M$ is represented on $L^2(M)$ as above, we have
(see \cite{K1})
\begin{equation}\label{F-A.1}
[D\omega:D\ffi]_t=h_\omega^{it}h_\ffi^{-it},\qquad t\in\bR.
\end{equation}

\subsection{Lemmas}
For later use we state the following two lemmas, while it seems that they are known to
specialists in the subject matter. The first lemma generalizes \cite[Theorem 3]{C1} and
\cite[Lemma 3.13]{Co}. These lemmas can be shown by \eqref{F-A.1} and \cite[\S9.24]{StZs}
together with a usual argument in analytic function theory (see also \cite{K4}).

\begin{lemma}\label{L-A.1}
For each $\ffi,\omega\in M_*^+$ and $\delta>0$ the following conditions are equivalent:
\begin{itemize}
\item[(i)] $h_\omega^\delta\le\mu h_\ffi^\delta$, i.e.,
$\mu h_\ffi^\delta-h_\omega^\delta\in L^\delta(M)_+$ for some $\mu>0$;
\item[(ii)] $s(\omega)\le s(\ffi)$ and $[D\omega:D\ffi]_t$ extends to a weakly continuous
($M$-valued) function $[D\omega:D\ffi]_z$ on the strip $-\delta/2\le\Im z\le0$ which is
analytic in the interior.
\end{itemize}

If the above conditions hold, then $\|[D\omega:D\ffi]_z\|\le\mu^{1/2}$ and $[D\omega:D\ffi]_z$
is strongly continuous on $-\delta/2\le\Im z\le0$, and
$$
h_\omega^{p/2}=[D\omega:D\ffi]_{-ip/2}h_\ffi^{p/2},\qquad0<p\le\delta.
$$
\end{lemma}

\begin{lemma}\label{L-A.2}
For each $\ffi,\omega\in M_*^+$ and $\delta>0$ the following conditions are equivalent:
\begin{itemize}
\item[(i)] $\mu^{-1}h_\ffi^\delta\le h_\omega^\delta\le\mu h_\ffi^\delta$ for some $\mu>0$;
\item[(ii)] $s(\omega)=s(\ffi)$ and $[D\omega:D\ffi]_t$ extends to a weakly continuous
($M$-valued) function $[D\omega:D\ffi]_z$ on the strip $-\delta/2\le\Im z\le\delta/2$ which is
analytic in the interior.
\end{itemize}

If the above conditions hold, then $\|[D\omega:D\ffi]_z\|\le\mu^{1/2}$,
$$
[D\omega:D\ffi]_{\bar z}^*=[D\ffi:D\omega]_z=[D\omega:D\ffi]_z^{-1},
$$
and $[D\omega:D\ffi]_z$ is strongly* continuous on
$-\delta/2\le\Im z\le\delta/2$.
\end{lemma}

For every $\rho,\sigma\in M_*^+$ and $x\in M$ recall the following well-known identity
\begin{align}\label{F-A.2}
\Delta_{\rho,\sigma}^p(xh_\sigma^{1/2})=h_\rho^pxh_\sigma^{1/2-p},
\qquad0\le p\le1/2,
\end{align}
with the convention that $h_\rho^0=s(\rho)$, $h_\sigma^0=s(\sigma)$ and
$\Delta_{\rho,\sigma}^0=s(\rho)Js(\sigma)J$. Indeed, this is seen from the uniqueness of
analytic continuation of $\Delta_{\rho,\sigma}^{it}(xh_\sigma^{1/2})
=h_\rho^{it}xh_\sigma^{1/2-it}$ for $t\in\bR$ (see \cite{K1}).

Another lemma we need is the following:

\begin{lemma}\label{L-A.3}
For every $\rho,\sigma\in M_*^+$ and $0\le p\le1/2$, the domain of $\Delta_{\rho,\sigma}^p$ is
\begin{align}\label{F-A.3}
\cD(\Delta_{\rho,\sigma}^p)=\bigl\{\xi\in L^2(M):
h_\rho^p\xi s(\sigma)=\eta h_\sigma^p\ \mbox{for some}\ \eta\in L^2(M)\bigr\}.
\end{align}
Moreover, if $\xi,\eta\in L^2(M)$ are given as in \eqref{F-A.3}, then
\begin{align}\label{F-A.4}
\Delta_{\rho,\sigma}^p\xi=\Delta_{\rho,\sigma}^p(\xi s(\sigma))=\eta s(\sigma).
\end{align}
\end{lemma}

\begin{proof}
Assume first that $\xi,\eta\in L^2(M)$ are given with $h_\rho^p\xi s(\sigma)=\eta h_\sigma^p$
and so $s(\sigma)\xi^*h_\rho^p=h_\sigma^p\eta^*$. For every $x\in M$ we have by \eqref{F-A.2}
\begin{align*}
\<\xi,\Delta_{\rho,\sigma}^p(xh_\sigma^{1/2})\>
&=\<\xi,h_\rho^pxh_\sigma^{1/2-p}\>=\tr(s(\sigma)\xi^*h_\rho^pxh_\sigma^{1/2-p}) \\
&=\tr(h_\sigma^p\eta^*xh_\sigma^{1/2-p})=\tr(s(\sigma)\eta^*xh_\sigma^{1/2})
=\<\eta s(\sigma),xh_\sigma^{1/2}\>.
\end{align*}
This equality immediately extends to
$$
\<\xi,\Delta_{\rho,\sigma}^p\zeta\>=\<\eta s(\sigma),\zeta\>,\qquad
\zeta\in Mh_\sigma^{1/2}+L^2(M)(1-s(\sigma)).
$$
Since $Mh_\sigma^{1/2}+L^2(M)(1-s(\sigma))$ is a core of $\Delta_{\rho,\sigma}^{1/2}$, it is
also a core of $\Delta_{\rho,\sigma}^p$ (since $0<p\le1/2$) by \cite[Lemma 4]{Ar1}. Hence
we find that $\xi\in\cD(\Delta_{\rho,\sigma}^p)$ and \eqref{F-A.4} holds.

Conversely, assume that $\xi\in\cD(\Delta_{\rho,\sigma}^p)$ and
$\Delta_{\rho,\sigma}^p\xi=\eta\in L^2(M)$, so $\Delta_{\rho,\sigma}^p(\xi s(\sigma))=\eta$
and $\eta s(\sigma)=\eta$. Since $Mh_\sigma^{1/2}+L^2(M)(1-s(\sigma))$ is a core of
$\Delta_{\rho,\sigma}^p$, there exists a sequence $\{x_n\}$ in $M$ such that
$$
\|x_nh_\sigma^{1/2}-\xi s(\sigma)\|\,\longrightarrow\,0,\qquad
\|\Delta_{\rho,\sigma}^p(x_nh_\sigma^{1/2})-\eta\|\,\longrightarrow\,0.
$$
Let $\eta_n:=\Delta_{\rho,\sigma}^p(x_nh_\sigma^{1/2})$; then
$\eta_n=h_\rho^px_nh_\sigma^{1/2-p}$ by \eqref{F-A.2}. We hence have
\begin{align}\label{F-A.5}
\eta_nh_\sigma^p=h_\rho^px_nh_\sigma^{1/2}.
\end{align}
By H\"older's inequality \cite{Te},
\begin{align}
\|\eta_nh_\sigma^p-\eta h_\sigma^p\|_{2p\over p+2}
&\le\|\eta_n-\eta\|_2\|h_\sigma^p\|_p\,\longrightarrow\,0, \label{F-A.6}\\
\|h_\rho^px_nh_\sigma^{1/2}-h_\rho^p\xi s(\sigma)\|_{2p\over p+2}
&\le\|h_\rho^p\|_p\|x_nh_\sigma^{1/2}-\xi s(\sigma)\|_2\,\longrightarrow\,0. \label{F-A.7}
\end{align}
Combining \eqref{F-A.5}--\eqref{F-A.7} yields $h_\rho^p\xi s(\sigma)=\eta h_\sigma^p$. Thus
\eqref{F-A.3} follows.
\end{proof}

\subsection{Description of R\'enyi divergences}
The following provides a useful description of the R\'enyi divergence $D_\alpha(\rho\|\sigma)$
in terms of $h_\rho,h_\sigma\in L^1(M)_+$.

\begin{prop}\label{P-A.4}
Let $\rho,\sigma\in M_*^+$.
\begin{itemize}
\item[(1)] When $0\le\alpha<1$,
\begin{align}\label{F-A.8}
Q_\alpha(\rho\|\sigma)=\tr(h_\rho^\alpha h_\sigma^{1-\alpha}).
\end{align}
\item[(2)] When $s(\rho)\le s(\sigma)$ and $1<\alpha\le2$, the following conditions are
equivalent:
\begin{itemize}
\item[(i)] $h_\sigma^{1/2}\in\cD(\Delta_{\rho,\sigma}^{\alpha/2})$;
\item[(ii)] $h_\rho^{1/2}\in\cD(\Delta_{\rho,\sigma}^{(\alpha-1)/2})$;
\item[(iii)] there exists an $\eta\in L^2(M)s(\sigma)$ such that
$h_\rho^{\alpha/2}=\eta h_\sigma^{(\alpha-1)/2}$.
\end{itemize}

If the above conditions hold, then $\eta$ in (iii) is unique and
$Q_\alpha(\rho\|\sigma)=\|\eta\|_2^2$.
\end{itemize}
\end{prop}

\begin{proof}
(1)\enspace For $0<\alpha<1$ we have by \eqref{F-A.2}
$$
Q_\alpha(\rho\|\sigma)=\|\Delta_{\rho,\sigma}^{\alpha/2}h_\sigma^{1/2}\|^2
=\|h_\rho^{\alpha/2}h_\sigma^{(1-\alpha)/2}\|^2
=\tr(h_\sigma^{(1-\alpha)/2}h_\rho^\alpha h_\sigma^{(1-\alpha)/2})
=\tr(h_\rho^\alpha h_\sigma^{1-\alpha}).
$$

(2)\enspace
Assume that $s(\rho)\le s(\sigma)$ and $1<\alpha\le2$. Since
$h_\sigma^{1/2}\in\cD(\Delta_{\rho,\sigma}^{1/2})$ and
$\Delta_{\rho,\sigma}^{1/2}h_\sigma^{1/2}=h_\rho^{1/2}s(\sigma)=h_\rho^{1/2}$, it follows
that (i)\,$\iff$\,(ii) and in this case $\Delta_{\rho,\sigma}^{\alpha/2}h_\sigma^{1/2}
=\Delta_{\rho,\sigma}^{(\alpha-1)/2}h_\rho^{1/2}$ (see \cite[Theorem 9.20]{StZs} for details).
Hence Lemma \ref{L-A.3} with $p=(\alpha-1)/2$ implies that
(ii)\,$\iff$\,(iii) and in this case $Q_\alpha(\rho\|\sigma)=\|\eta\|_2^2$. The uniqueness of
$\eta$ in (iii) is obvious.
\end{proof}

\begin{remark}\label{R-A.5}\rm
In the above (2), if $h_\rho^{\alpha/2}=\eta h_\sigma^{(\alpha-1)/2}$ with
$\eta\in L^2(M)s(\sigma)$, then we may write $\eta=h_\rho^{\alpha/2}h_\sigma^{(1-\alpha)/2}$
in a formal sense, so that
$$
Q_\alpha(\rho\|\sigma)=\|h_\rho^{\alpha/2}h_\sigma^{(1-\alpha)/2}\|^2
=\tr(h_\rho^{\alpha/2}h_\sigma^{1-\alpha}h_\rho^{\alpha/2})
=\tr(h_\rho^\alpha h_\sigma^{1-\alpha}),
$$
which is the same expression as in (1). This is in the same form as the quantity $Q_\alpha$
in the matrix case if we consider $\tr$ as the usual trace and $h_\rho,h_\sigma$ as the density
matrices.
\end{remark}

The next proposition gives a strengthening of Proposition \ref{P-5.3}\,(8).

\begin{prop}\label{P-A.6}
Let $\rho_i,\sigma_i\in M_*^+$ for $i=1,2$, and $\mu>0$.
\begin{itemize}
\item[(1)] Assume that $0\le\alpha<1$. Then
$h_{\sigma_1}^{1-\alpha}\le\mu h_{\sigma_2}^{1-\alpha}$ if and only if
$Q_\alpha(\rho\|\sigma_1)\le\mu Q_\alpha(\rho\|\sigma_2)$ for all $\rho\in M_*^+$.
\item[(2)] Assume that $1<\alpha\le2$. If $h_{\rho_1}^\alpha\le\mu h_{\rho_2}^\alpha$, then
$Q_\alpha(\rho_1\|\sigma)\le\mu Q_\alpha(\rho_2\|\sigma)$ for all $\sigma\in M_*^+$. If
$h_{\sigma_1}^{\alpha-1}\le\mu h_{\sigma_2}^{\alpha-1}$, then
$Q_\alpha(\rho\|\sigma_1)\ge\mu^{-1}Q_\alpha(\rho\|\sigma_2)$ for all $\rho\in M_*^+$.
\end{itemize}
\end{prop}

\begin{proof}
(1)\enspace
Let $0\le\alpha<1$. Recall \cite[Proposition II.33]{Te} that for any $b\in L^{1/(1-\alpha)}(M)$,
$b\ge0$ if and only if $\tr(ab)\ge0$ for all $a\in L^{1/\alpha}(M)_+$ (with $1/\alpha=\infty$
for $\alpha=0$). Hence the assertion is immediate from \eqref{F-A.8}.

(2)\enspace
Let $1<\alpha\le2$. Assume that $h_{\rho_1}^\alpha\le\mu h_{\rho_2}^\alpha$. To prove the
asserted inequality, we may assume that $Q_\alpha(\rho_2\|\sigma)<+\infty$ so that
$s(\rho_2)\le s(\sigma)$ and $h_\sigma^{1/2}\in\cD(\Delta_{\rho_2,\sigma}^{\alpha/2})$. By
Proposition \ref{P-A.4}\,(2) there exists an $\eta\in L^2(M)s(\sigma)$ such that
$h_{\rho_2}^{\alpha/2}=\eta h_\sigma^{(\alpha-1)/2}$. By Lemma \ref{L-A.1} one has
$a:=[D\rho_1:D\rho_2]_{-i\alpha/2}\in M$, for which
$h_{\rho_1}^{\alpha/2}=ah_{\rho_2}^{\alpha/2}$ and $\|a\|\le\mu^{1/2}$. Therefore,
$h_{\rho_1}^{\alpha/2}=a\eta h_\sigma^{(\alpha-1)/2}$, so Proposition \ref{P-A.4}\,(2) gives
$Q_\alpha(\rho_1\|\sigma)=\|a\eta\|_2^2\le\mu\|\eta\|_2^2=\mu Q_\alpha(\rho_2\|\sigma)$.

Next, assume that $h_{\sigma_1}^{\alpha-1}\le\mu h_{\sigma_2}^{\alpha-1}$. We may assume that
$Q_\alpha(\rho\|\sigma_1)<+\infty$ so that
$s(\rho)\le s(\sigma_1)$ and $h_\sigma^{1/2}\in\cD(\Delta_{\rho,\sigma_1}^{\alpha/2})$. By
Proposition \ref{P-A.4}\,(2) one has $h_\rho^{\alpha/2}=\eta h_{\sigma_1}^{(\alpha-1)/2}$ for
some $\eta\in L^2(M)s(\sigma_1)$. By Lemma \ref{L-A.1} one has
$h_{\sigma_1}^{(\alpha-1)/2}=bh_{\sigma_2}^{(\alpha-1)/2}$ for some $b\in Ms(\sigma_2)$ with
$\|b\|\le\mu^{1/2}$. Therefore, $h_\rho^{\alpha/2}=\eta bh_{\sigma_2}^{(\alpha-1)/2}$ so that
$Q_\alpha(\rho\|\sigma_2)=\|\eta b\|_2^2\le\mu\|\eta\|_2^2=\mu Q_\alpha(\rho\|\sigma_1)$.
\end{proof}

\begin{remark}\label{R-A.7}\rm
Anna Jen\v cov\'a \cite{Je3} informed the author that she could prove Lemma \ref{L-A.3} for
every $p\ge0$ by using analyticity of $z\mapsto h_\rho^z$ in $\Re z>0$ \cite[Lemma II.18]{Te}
and a convergence argument. Then the equivalence of (i)--(iii) in Proposition \ref{P-A.4}\,(2)
is true for all $\alpha>1$, thus Propositions \ref{P-A.6}\,(2) holds for all $\alpha>1$. 
\end{remark}

\section{Relative hamiltonians in terms of Haagerup's $L^p$-spaces}

The theory of relative hamiltonians in the general von Neumann algebra setting was formerly
developed by Araki \cite{Ar3} and Donald \cite{Do} in close relation to the relative entropy.
Although the topic is not strongly related to the main body of this paper, it is worthwhile
to consider that in a more general framework in terms of Haagerup's $L^p$-spaces, as a sequel
of Appendix~A.

\subsection{Survey on relative hamiltonians}

Let $(M,\cH,J,\cP)$ be a standard form of $M$ and $\ffi\in M_*^+$ be faithful so that
$\ffi=\<\Phi,\,\cdot\,\Phi\>$ with the cyclic and separating vector $\Phi\in\cP$. For each
$h\in M_\sa$ Araki \cite{Ar3} defined the perturbed vector $\Phi^h$ by
\begin{align}\label{F-B.1}
\Phi^h:=\sum_{n=0}^\infty\int_0^{1/2}dt_1\int_0^{t_1}dt_2\cdots\int_0^{t_{n-1}}dt_n
\Delta_\ffi^{t_n}h\Delta_\ffi^{t_{n-1}-t_n}h\cdots\Delta_\ffi^{t_1-t_2}h\Phi,
\end{align}
where $\Phi$ is in the domain of $\Delta_\ffi^{z_1}h\Delta_\ffi^{z_2}h\cdots\Delta_\ffi^{z_n}h$
for $z=(z_1,\dots,z_n)\in\bC^n$ with
$\Re z\in\{(s_1,\dots,s_n)\in\bR^n:s_1,\dots,s_n\ge0,\,s_1+\cdots+s_n\le1/2\}$ and the
right-hand side of \eqref{F-B.1} absolutely converges. Then $\Phi^h$ is also a cyclic and
separating vector in $\mathcal{P}$. It is known \cite{Ar4} that $\Phi$ is in the domain of
$\exp{1\over2}(\log\Delta_\ffi+h)$ and $\Phi^h=\exp{1\over2}(\log\Delta_\ffi+h)\Phi$. A version
of the Trotter product formula was given in \cite[Remarks 1, 2]{Ar4} as follows:
$$
\Phi^h=\lim_{n\to\infty}(\Delta_\ffi^{1/2n}e^{h/2n})^n\Phi
=\lim_{n\to\infty}(e^{h/2n}\Delta_\ffi^{1/2n})^n\Phi\ \ \mbox{(strong limit)}.
$$
The perturbed functional $\ffi^h$ is defined by $\ffi^h=\<\Phi^h,\,\cdot\,\Phi^h\>$
independently of the choice of a standard form. In fact, we have (see
\cite[Proposition 4.3]{Ar3} and the proof of \cite[Theorem 6]{Pe1})
\begin{align}\label{F-B.2}
[D\ffi^h:D\ffi]_t
=\sum_{n=0}^\infty i^n\int_0^tdt_1\int_0^{t_1}dt_2\cdots\int_0^{t_{n-1}}dt_n
\sigma^\ffi_{t_n}(h)\cdots\sigma^\ffi_{t_1}(h).
\end{align}
When $\omega=\ffi^h$ with $h\in M_\sa$, $-h$ is called a \emph{relative hamiltonian} of
$\omega$ relative to $\ffi$.

Let $\fS(M)$ denote the set of normal states on $M$. If $h:\fS(M)\to(-\infty,+\infty]$ is a
weakly lower semicontinuous affine map whose range is lower bounded, then $h$ is called an
\emph{extended-valued lower-bounded operator} affiliated with $M$. We denote by $M_\ext$
the set of all such extended-valued operators. Note \cite[Proposition 2.13]{Do} (also
\cite[Theorem 1.5]{Ha3}) that for each $h\in M_\ext$ there exist a projection $p\in M$ and
a spectral resolution $\{e_\lambda:s\le\lambda<\infty\}$ in $M$ with $s\in\bR$ and
$e_\infty=1-p$ such that $h$ is represented as
$$
h(\rho)=\int_s^\infty\lambda\,d\rho(e_\lambda)+\infty\rho(p),\qquad\rho\in\fS(M),
$$
and that $h\in M_\ext$ if and only if there exists an increasing net $\{h_\alpha\}$ in $M_\sa$
such that $h(\rho)=\sup_\alpha\rho(h_\alpha)$ for $\rho\in\fS(M)$. Obviously, $h\in M_\ext$
can extend to a positively homogeneous map $h:M_*^+\to(-\infty,+\infty]$. We further set
$-M_\ext:=\{-h:h\in M_\ext\}$. Namely, $h\in-M_\ext$ is represented as
$$
h(\rho)=\int_{-\infty}^r\lambda\,d\rho(e_\lambda)+(-\infty)\rho(p),
\qquad\rho\in M_*^+,
$$
where $\{e_\lambda:-\infty<\lambda\le r\}$ is a spectral resolution in $M$ with $r\in\bR$ and
$e_r=1-p$. In this case we write $h\le r$. For each $\ffi\in M_*^+$ and $h\in M_\ext$ define
$c(\ffi,h)\in(-\infty,+\infty]$ by
$$
c(\ffi,h):=\inf\{h(\rho)+D(\rho\|\ffi):\rho\in\fS(M)\}.
$$
Donald \cite[Theorem 3.1]{Do} proved that if $c(\ffi,h)<+\infty$, that is,
$h(\rho)+D(\rho\|\ffi)<\infty$ for some $\rho\in\fS(M)$, then there exists a unique
$\omega\in\fS(M)$ such that $h(\omega)+D(\omega\|\ffi)=c(\ffi,h)$. We denote this $\omega$ by
$[\ffi^h]$; this was denoted in \cite{Do} by $\ffi^h$ but we prefer to use notation
consistently with Araki's one.

Note \cite[Proposition 1]{Pe2} (also \cite[Appendix]{PRV}) that when $\ffi\in M_*^+$ is
faithful and $h\in M_\sa$, $[\ffi^h]$ coincides with $\ffi^{-h}$ up to a normalization
constant; more precisely, $\ffi^{-h}=\ffi^{-h}(1)[\ffi^h]$ and $\ffi^{-h}(1)=e^{-c(\ffi,h)}$.
So, for any $\ffi\in M_*^+$ and $h\in-M_\ext$ such that $c(\ffi,-h)<+\infty$ we can define
the perturbed functional $\ffi^h$ by
$$
\ffi^h:=e^{-c(\ffi,-h)}[\ffi^{-h}]
$$
and call $-h$ a \emph{relative hamiltonian} of $\omega=\ffi^h$ relative to $\ffi$. But when
$c(\ffi,-h)=\infty$ let $\ffi^h:=0$ as convention. Moreover, it is worthwhile to recall Petz'
variational expression of the relative entropy \cite[Theorem 9]{Pe2}; namely, if $\ffi$ is faithful and
$\omega\in\fS(M)$, then
\begin{align}\label{F-B.3}
D(\omega\|\ffi)=\sup\{\omega(h)-\log\ffi^h(1):h\in M_\sa\}.
\end{align}
(Even when $\ffi$ is not faithful, \eqref{F-B.3} remains valid with $\omega(s(\ffi)hs(\ffi))$
in place of $\omega(h)$.)

\subsection{Theorems}

We here prove the next theorems as to the existence of relative hamiltonian, generalizing
\cite[Theorem 6.3]{Ar3} and \cite[Theorem 4.3]{Do}, respectively, in the framework of the
standard form $(M,L^2(M),\,^*\,,L^2(M)_+)$.

\begin{thm}\label{T-B.1}
If $\ffi,\omega\in M_*^+$ and $\nu h_\ffi^\delta\le h_\omega^\delta\le\mu h_\ffi^\delta$
for some $\delta,\mu,\nu>0$, then there exists an $h\in M_\sa$ such that $\omega=\ffi^h$ and
$\delta^{-1}\log\nu\le h\le\delta^{-1}\log\mu$.
\end{thm}

\begin{thm}\label{T-B.2}
If $\ffi,\omega\in M_*^+$ and $h_\omega^\delta\le\mu h_\ffi^\delta$ for some $\delta,\mu>0$,
then there exists an $h\in-M_\ext$ such that $\omega=\ffi^h$, $h\le\delta^{-1}\log\mu$, and
\begin{equation}\label{F-B.4}
D(\rho\|\omega)=-h(\rho)+D(\rho\|\ffi),\qquad\rho\in M_*^+.
\end{equation}
\end{thm}

To prove the theorems, we first give the following:

\begin{lemma}\label{L-B.3}
If $\ffi,\omega\in M_*^+$ and $h_\omega^\delta\le\mu h_\ffi^\delta$ for some $\delta,\mu>0$,
then
$$
D(\rho\|\omega)\ge D(\rho\|\ffi)-\rho(1){\log\mu\over\delta},\qquad\rho\in M_*^+.
$$
\end{lemma}

\begin{proof}
Since both sides of the asserted inequality are zero if $\rho=0$, we may assume that $\rho\ne0$.
By assumption we have $h_\omega^{1-\alpha}\le\mu^{(1-\alpha)/\delta}h_\ffi^{1-\alpha}$ for any
$\alpha\in(0,1)$ with $1-\alpha\le\delta$. Then Proposition \ref{P-A.6}\,(1) implies that
$Q_\alpha(\rho\|\omega)\le\mu^{(1-\alpha)/\delta}Q_\alpha(\rho\|\ffi)$ so that
$$
\log{Q_\alpha(\rho\|\omega)\over\rho(1)}\le\log{Q_\alpha(\rho\|\ffi)\over\rho(1)}
+{1-\alpha\over\delta}\log\mu.
$$
Therefore, $D_\alpha(\rho\|\omega)\ge D_\alpha(\rho\|\ffi)-\delta^{-1}\log\mu$. Letting
$\alpha\nearrow1$ gives the desired inequality due to Proposition \ref{P-5.3}\,(3).
\end{proof}

\noindent
{\it Proof of Theorem \ref{T-B.1}.}\enspace
Let $\ffi,\omega\in M_*^+$ satisfy the assumption of the theorem. We may suppose that
$\ffi$ is faithful (hence so is $\omega$). By Lemma \ref{L-A.2} and \cite[\S9.24]{StZs}
one can define $h\in M$ by $h=-i\,{d\over dz}[D\omega:D\ffi]_z|_{z=0}$. Then since
$[D\omega:D\ffi]_t$ is a $\sigma_t^\ffi$-unitary cocycle (see \cite{St,Ta}), one has
$h\in M_\sa$ and
$$
{d\over dt}[D\omega:D\ffi]_t=[D\omega:D\ffi]_t\sigma_t^\ffi(ih),\qquad t\in\bR,
$$
so that $[D\omega:D\ffi]_t=[D\ffi^h:D\ffi]_t$ for all $t\in\bR$ by \cite[Theorem 1]{Ar0}
and \eqref{F-B.2}. Therefore $\omega=\ffi^h$. For every $\rho\in M_*^+$ one has by
\cite[Theorem 3.10]{Ar2}
\begin{equation}\label{F-B.5}
D(\rho\|\omega)=-\rho(h)+D(\rho\|\ffi)
\end{equation}
and by Lemma \ref{L-B.3}
\begin{equation}\label{F-B.6}
D(\rho\|\omega)\ge D(\rho\|\ffi)-\rho(1){\log\mu\over\delta},
\end{equation}
\begin{equation}\label{F-B.7}
D(\rho\|\ffi)\ge D(\rho\|\omega)+\rho(1){\log\nu\over\delta}.
\end{equation}
When $D(\rho\|\ffi)<\infty$, estimates \eqref{F-B.5}--\eqref{F-B.7} imply that
$$
\rho(1){\log\nu\over\delta}\le\rho(h)\le\rho(1){\log\mu\over\delta}.
$$
But the set of $\rho\in M_*^+$ with $D(\rho\|\ffi)<\infty$ is dense in $M_*^+$ by the faithfulness of $\ffi$. Hence $\delta^{-1}\log\nu\le h\le\delta^{-1}\log\mu$.\qed

\bigskip\noindent
{\it Proof of Theorem \ref{T-B.2}.}\enspace
Assume that $\ffi,\omega\in M_*^+$ and $h_\omega^\delta\le\mu h_\ffi^\delta$ with
$\delta,\mu>0$. We may suppose that $\ffi$ is faithful and $\omega$ is nonzero (the case
$\omega=0$ is trivial). For each $\eps>0$ define $\omega(\eps)\in M_*^+$ by
$h_{\omega(\eps)}=(h_\omega^\delta+\eps h_\ffi^\delta)^{1/\delta}$. Since
$\eps h_\ffi^\delta\le h_{\omega(\eps)}^\delta\le(\mu+\eps)h_\ffi^\delta$, Theorem \ref{T-B.1}
implies that there exists an $h(\eps)\in M_\sa$ such that $\omega(\eps)=\ffi^{h(\eps)}$ and
$h(\eps)\le\delta^{-1}\log(\mu+\eps)$. When $0<\eps<\eps'$, it follows from Lemma \ref{L-B.3}
that $D(\rho\|\omega(\eps))\ge D(\rho\|\omega(\eps'))$ for all $\rho\in M_*^+$. Since by
\cite[Theorem 3.10]{Ar2}
\begin{align}\label{F-B.8}
D(\rho\|\omega(\eps))=-\rho(h(\eps))+D(\rho\|\ffi),\qquad\rho\in M_*^+,
\end{align}
one has $-\rho(h(\eps))\ge-\rho(h(\eps'))$ for all $\rho\in M_*^+$ with
$D(\rho\|\ffi)<\infty$, so that $-h(\eps)\ge-h(\eps')$ as in the proof of Theorem \ref{T-B.1}.
So $h\in-M_\ext$ can be defined by
$$
-h(\rho)=\sup_{\eps>0}\rho(-h(\eps))=\lim_{\eps\searrow0}\rho(-h(\eps)),
\qquad\rho\in M_*^+.
$$
Then one has $h\le\delta^{-1}\log\mu$ as the limit of $h(\eps)\le\delta^{-1}\log(\mu+\eps)$.

Next, let us prove \eqref{F-B.4}. Note \cite[Lemma II.16 and Corollary II.17]{Te} that
the norm topology on $L^p(M)$ ($\subset\widetilde N$) coincides with the measure topology
(induced by the canonical trace on $N$). We see by \cite{Ti} that
$h_{\omega(\eps)}\to h_\omega$ as $\eps\searrow0$ in the measure topology. Hence
$\|\omega(\eps)-\omega\|\to0$, so that we have for every $\rho\in M_*^+$
$$
D(\rho\|\omega)\le\liminf_{\eps\searrow0}D(\rho\|\omega(\eps))
$$
by the lower semicontinuity of relative entropy \cite[Theorem 3.7(1)]{Ar2}. On the other hand,
Lemma \ref{L-B.3} shows that $D(\rho\|\omega)\ge D(\rho\|\omega(\eps))$ for all $\eps>0$
because $h_\omega^\delta\le h^\delta_{\omega(\eps)}$. Therefore,
$$
D(\rho\|\omega)=\lim_{\eps\searrow0}D(\rho\|\omega(\eps)).
$$
Now \eqref{F-B.4} is immediate by taking the limit of \eqref{F-B.8}. Finally, \eqref{F-B.4}
implies that for every $\rho\in\fS(M)$
$$
-h(\rho)+D(\rho\|\ffi)=D(\rho\|\omega)\ge-\log\omega(1)
$$
and in particular, when $\rho=\omega/\omega(1)$,
$$
-h(\rho)+D(\rho\|\ffi)=D(\omega/\omega(1)\|\omega)=-\log\omega(1).
$$
Hence $c(\ffi,-h)=-\log\omega(1)<\infty$ and $\omega/\omega(1)=[\ffi^{-h}]$, so that
$\omega=e^{-c(\ffi,-h)}[\ffi^{-h}]=\ffi^h$, completing the proof.\qed

\begin{remark}\rm
Assume that $\ffi\in M_*^+$ be faithful and $\omega=\ffi^h$ for some $h\in M_\sa$.
Note \cite[Proposition 4.6]{Ar3} that $-h$ is a unique relative hamiltonian of $\omega$
relative to $\ffi$. This is seen also from \cite[Theorem 3.10]{Ar2} or more explicitly from
\cite[(4.28)]{Ar2}. Moreover, according to \cite[Theorem III.1]{Hs} and its remarks we have
$$
h_\omega^{1/2}=\exp\biggl({\log h_\ffi+h\over2}\biggr)
=\lim_{n\to\infty}(h_\ffi^{1/2n}e^{h/2n})^n\ \ \mbox{(weak limit)}
$$
and hence
\begin{align}\label{F-B.9}
h_\omega=\exp(\log h_\ffi+h).
\end{align}
(In fact, the results of \cite{Hs} in the spatial $L^p$-spaces can be automatically
transformed into those in the Haagerup $L^p$-spaces due to \cite[Theorem IV.12]{Te}.)
Since $D(\omega\|\ffi)=\omega(h)$, \eqref{F-B.9} yields the formula
\begin{align}\label{F-B.10}
D(\omega\|\ffi)=\tr(h_\omega(\log h_\omega-\log h_\ffi)),
\end{align}
which has a complete resemblance to Umegaki's relative entropy in the semifinite case (see
\eqref{F-1.1}). However, when $\omega=\ffi^h$ with a relative hamiltonian $-h\in M_\ext$
unbounded from above, it seems problematic to determine whether formulas \eqref{F-B.9} and
\eqref{F-B.10} remain to make sense (see \cite{Mu} for a related discussion).
\end{remark}

\end{document}